\documentclass[10pt,twocolumn,twoside]{IEEEtran}
\usepackage{blindtext}

\usepackage[english]{babel}

\usepackage{cite} % Orders citations.
\usepackage{url}
\usepackage{hyperref}

\usepackage{color}
\usepackage[cmex10]{amsmath}
\usepackage{amsfonts, amssymb, amsthm}
\usepackage{mathrsfs}
\usepackage{bm}

\usepackage{algorithm,algpseudocode}
	\algnewcommand{\LeftComment}[1]{\Statex \(\triangleright\) #1}

\usepackage{enumerate}
\usepackage{caption}
\usepackage{rotating}
\usepackage{subcaption}
\input{mysymbol.sty}

\newcommand{\yhat}{\hat{y}}

\newtheorem{theorem}{Theorem}
\newtheorem{proposition}{Proposition}
\newtheorem{corollary}{Corollary}
\newtheorem{lemma}{Lemma}
\theoremstyle{definition}

\newtheorem{remark}{Remark}
\newtheoremstyle{assume}
  {3pt}% measure of space to leave above the theorem. e.g.: 3pt
  {3pt}% measure of space to leave below the theorem. e.g.: 3pt
  {}% name of font to use in the body of the theorem
  {}% measure of space to indent
  {\bf}% name of head font
  {}% punctuation between head and body
  { }% space after theorem head; " " = normal interword space
  {\thmname{#1}.\thmnumber{#2}\thmnote{ \textnormal{(\textit{#3})}}}% Manually specify head
\theoremstyle{assume}

\def\nnil{\nil}
\newcounter{prob}
\newenvironment{prob}[1][\nil]{%
	\def\tmp{#1}
	\equation
	\ifx\tmp\nnil
		\refstepcounter{prob}
		\tag{P\Roman{prob}}
	\else
		\tag{\tmp}
	\fi
	\aligned%
}{%
	\endaligned\endequation%
}

\DeclareMathOperator*{\argmin}{argmin}

\DeclareMathOperator*{\minimize}{minimize}
\DeclareMathOperator*{\maximize}{maximize}
\DeclareMathOperator{\subjectto}{subject\ to}
\DeclareMathOperator{\indicator}{\mathbb{I}}
\DeclareMathOperator{\sspan}{span}

\RequirePackage[normalem]{ulem}
\RequirePackage{color}\definecolor{RED}{rgb}{1,0,0}\definecolor{BLUE}{rgb}{0,0,1}

\title{Sparse multiresolution representations with adaptive kernels}

\author{Maria~Peifer, Luiz~F.~O.~Chamon, Santiago~Paternain, and Alejandro~Ribeiro%
\thanks{Department of Electrical and Systems Engineering, University of Pennsylvania.
e-mail: \mbox{\texttt{mariaop@seas.upenn.edu}}~(contact author), \mbox{\texttt{\{luizf,spater,aribeiro\}@seas.upenn.edu}}. Their work is partially supported by the National Science Foundation CCF 1717120 and ARO W911NF1710438.}
\thanks{Part of the results in this paper previously appeared in~\cite{peifer2018locally, peifer2019sparse}.}%
}

\begin{document}

\maketitle

\begin{abstract}

Reproducing kernel Hilbert spaces~(RKHSs) are key elements of many non-parametric tools successfully used in signal processing, statistics, and machine learning. In this work, we aim to address three issues of the classical RKHS-based techniques. First, they require the RKHS to be known \emph{a priori}, which is unrealistic in many applications. Furthermore, the choice of RKHS affects the shape and smoothness of the solution, thus impacting its performance. Second, RKHSs are ill-equipped to deal with heterogeneous degrees of smoothness, i.e., with functions that are smooth in some parts of their domain but vary rapidly in others. Finally, the computational complexity of evaluating the solution of these methods grows with the number of data points, rendering these techniques infeasible for many applications. Though kernel learning, local kernel adaptation, and sparsity have been used to address these issues, many of these approaches are computationally intensive or forgo optimality guarantees. We tackle these problems by leveraging a novel integral representation of functions in RKHSs that allows for arbitrary centers and different kernels at each center. To address the complexity issues, we then write the function estimation problem as a sparse functional program that explicitly minimizes the support of the representation leading to low complexity solutions. Despite their non-convexity and infinite dimensionality, we show these problems can be solved exactly and efficiently by leveraging duality, and we illustrate this new approach in simulated and real data.

\end{abstract}


%%%%%%%%%%%%%%%%%%%%%%%%%%%%%%%
%%% SECTION : Introduction  %%%
%%%%%%%%%%%%%%%%%%%%%%%%%%%%%%%

\section{Introduction} \label{sec_intro}

%!TEX root = mkl.tex
%%%%%%%%%%%%%%%%%%%%%%%%%%%%%%%%%%%%%%%%%%%%%%%%%%%%%%%%%%%%%%%%%%%%%
%   Introduction
%%%%%%%%%%%%%%%%%%%%%%%%%%%%%%%%%%%%%%%%%%%%%%%%%%%%%%%%%%%%%%%%%%%%%%

Reproducing kernel Hilbert spaces~(RKHSs) are at the core of non-parametric techniques in signal processing, statistics, and machine learning~\cite{scholkopf2001learning, bishop2006pattern, hofmann2008kernel, yuan2010reproducing, berlinet2011reproducing, Jeronimo13k, koppel2017parsimonious}. Their success stems from the fact that they combine the versatility of functional spaces with the tractability of parametric methods. Indeed, despite the richness of functions found in RKHSs, they can be represented as a~(possibly infinite) linear combination of simple basis functions called \emph{reproducing kernels}~(or Mercer kernels)~\cite{berlinet2011reproducing, bishop2006pattern}. For smooth functions, however, a celebrated variational result known as the \emph{representer theorem} states that the number of basis is finite and that they are given by the kernels evaluated at the data points~\cite{kimeldorf1971some, scholkopf2001generalized}. In other words, learning smooth functions in RKHSs is effectively a parametric, finite dimensional problem. Despite their success, RKHS learning methods have two major practical drawbacks: (i)~the RKHS must be known \emph{a priori} and (ii)~evaluating the learned function can be computationally prohibitive.

Indeed, the appropriate functional space of the solution is seldom known in practice. Moreover, since the RKHS dictates the shape and smoothness properties of its functions, its choice is application-specific and ultimately affects the learning performance~\cite{lanckriet2004learning, hofmann2008kernel, bergstra2012random, li2012automatic, kuo2014kernel}. Although there exist classes of kernels~(and thus RKHSs) that can approximate continuous function arbitrarily well~\cite{micchelli2005learning}, they may need a large number of data points to do so. In fact, it is well-known that RKHS methods are not sample efficient for learning functions with varying degrees of smoothness, i.e., functions that are smooth in some parts of their domain but vary rapidly in others~\cite{Donoho98m}.

Kernel learning approaches have been proposed to address this issue by fitting a combination of kernels from a predefined set~\cite{lanckriet2004learning, micchelli2005learning, gonen2011multiple} or by using spectral representations of positive-definite functions~\cite{ong2005learning, wilson2013gaussian, yang2015learning}. Even when the general form of the kernel is known, the choice of smoothness parameter remains~(e.g., selecting the bandwidth of Gaussian kernels). In this case, common approaches include grid search with cross-validation~\cite{bergstra2012random, kuhn2016applied} or application-specific heuristic such as maximizing the margin of the support-vectors machine~(SVM)~\cite{li2012automatic, kuo2014kernel}. These methods, however, quickly become impractical as they often must search over fine grids, use a large set of kernels, or require additional data. They also fail to address the issue of estimating functions with heterogeneous degrees of smoothness. Although this can be done by choosing different RKHSs for different regions of the domain by means of plug-in rules~\cite{Brockmann93l}, binary optimization~\cite{Liu14s}, hypothesis testing~\cite{ghosh2008kernel}, or gradient descent and alternating minimization~\cite{yuan2009adaptive, chen2016kernel}, these solutions come with no optimality guarantee due to the non-convexity of these locally adapted smoothness formulations.

The issue of computational complexity stems directly from the use of the representer theorem~\cite{kimeldorf1971some, scholkopf2001generalized, argyriou2009there}. Despite the reduction from the infinite dimensional problem of estimating smooth functions in RKHSs to that of estimating a finite number of parameters, solving of the problem can become computationally expensive because the number of parameters is proportional to the number of data points. Thus, evaluating the function at any point requires as many kernel evaluations as training samples, which in many applications is prohibitively high. This issue is often addressed by imposing a sparsity penalty on the coefficients to reduce the number of kernel evaluations. Greedy heuristics~\cite{smola2000sparse, vincent2002kernel} and~$\ell_1$-norm relaxations \cite{tibshirani1996regression, fung2002minimal, jud2016sparse, gao2013sparse}~are then often used to cope with the combinatorial nature of this problem, which is known to be NP-hard in general~\cite{Natarajan95s, amaldi1998approximability}. These methods, however, often implicitly rely on the classical representer theorems~\cite{kimeldorf1971some, scholkopf2001generalized}, despite the fact that they no longer hold in the presence of sparsity penalties~(see Remark~\ref{R:sparsity}). Hence, even if the sparse optimization program could be solved exactly, the solution would remain suboptimal with respect to the original function estimation problem.

In this work, we propose to tackle these issues by simultaneously (i)~adapting the kernels~(RKHSs) locally and (ii)~allowing arbitrary centers instead of constraining the kernels to be evaluated at the training samples. To do so, we put forward an integral representation of functions that lie in the sum space of an uncountable number of RKHSs taken from a parametrized family. This representation does not assume that the kernels are evaluated at any particular points of the domain. Enforcing sparsity on the coefficients of this representation allows us to determine the optimal parameter and center of each kernel locally. Moreover, sparsity has the added benefit of minimizing the number of kernels used. Despite the non-convex~(sparsity) and infinite dimensional~(functional) nature of the resulting optimization problem, we can leverage the results from~\cite{Chamon18s} to solve it exactly and efficiently using duality. These results also give rise to a new sparse representer theorem.

The paper is structured as follows: Section \ref{sec:problem} gives an overview of reproducing kernel Hilbert spaces and defines the problems associated with learning in these spaces. In section \ref{sec:solution} we formulate the problem using an integral representation of a function in an RKHS. In Section \ref{sec:dual} we show the solution to our problem in the dual domain, prove that the duality gap between the primal and the dual  problem is zero, and present our algorithm. Numerical examples are then used to illustrate the effectiveness of this method at locally identifying the correct kernel parameters as well as the kernel centers in different applications~(Section~\ref{sec:sims}).

%%%%%%%%%%%%%%%%%%%%%%%%%%%%%%%%%%%%%%%
%%% SECTION : Problem Formulation   %%%
%%%%%%%%%%%%%%%%%%%%%%%%%%%%%%%%%%%%%%%

\section{Learning in RKHSs: The Classical Problem} \label{sec:problem}

%!TEX root = mkl.tex
%%%%%%%%%%%%%%%%%%%%%%%%%%%%%%%%%%%%%%%
%%% SECTION : Problem Formulation   %%%
%%%%%%%%%%%%%%%%%%%%%%%%%%%%%%%%%%%%%%%

Given a training set of data pairs~$(\bbx_i, y_i)$, $i = 1, \dots, N$, where~$\bbx_i \in \ccalX$ are the observations or independent variables, with~$\ccalX \subset \reals^p$ compact, and~$y_i \in \reals$ is the label or dependent variable, we seek a function~$f: \ccalX \to \reals$ in the RKHS~$\ccalH_0$ that fits these data, i.e., such that~$c(f(\bbx_i),y_i)$ is small for some convex figure of merit~$c$, e.g., quadratic loss, hinge loss, or logistic log-likelihood. Formally, an RKHS is a complete, linear function space endowed with a unique reproducing kernel. A reproducing kernel~$k: \reals^p \times \reals^p \to \reals$ is a positive-semidefinite function with the property~$\langle f(\cdot), k(\cdot,\bbx) \rangle_{\ccalH_0} = f(\bbx)$ for any function~$f \in \ccalH_0$ and point~$\bbx \in \reals^p$, where~$\langle \cdot, \cdot \rangle_{\ccalH_0}$ denotes the inner product of the Hilbert space~$\ccalH_0$~\cite{berlinet2011reproducing}. Reproducing kernels are often parametrized by a constant that characterizes the smoothness/richness of the RKHS, such as the bandwidth of sincs, the order of polynomial kernels, or the scale/variance of Gaussian~(radial basis function, RBF) kernels~\cite{scholkopf2001learning, bishop2006pattern}. Just as each RKHS has a unique kernel, each kernel defines a unique RKHS. Explicitly, every function~$f \in \ccalH_0$ can be written as the pointwise limit of a linear combination of kernels, i.e.,
\begin{equation}\label{eqn_orig_repres}
	f(\bbx) = \lim_{n \to \infty} \sum_{j=1}^n a_j k(\bbx, \bbz_j \,;\, w_0)
		\text{,}
\end{equation}
where~$w_0$ denotes the kernel parameter and the~$\bbz_j \in \ccalX$ are called the kernel \emph{centers}. In other words, the RKHS~$\ccalH_0$ is the completion of the space~$\sspan\{k(\cdot, \bbz \,;\, w_0) \mid \bbz \in \ccalX\}$~\cite[Sec.~2]{aronszajn1950theory}. Note that different parameters~$w_0$ result in different RKHSs.

There are infinitely many representations of the form~\eqref{eqn_orig_repres} that can interpolate a finite set of points. To avoid overfitting the data and obtain a unique solution, a smoothness prior is often used by minimizing the RKHS norm of the solution~\cite{scholkopf2001learning, bishop2006pattern} as in 
\begin{prob}\label{eqn_funct_prob}
	\minimize_{f \in \ccalH_0}& &&\Vert f \Vert_{\ccalH_0}
	\\
 	\subjectto& &&c(f(\bbx_i), y_i) \leq 0
 		\text{,} \quad i = 1, \dots, N
 		\text{.}
\end{prob}
Despite its infinite dimensional nature, this optimization problem admits a solution written as a finite linear combination of kernels centered at the sample points, i.e., there exists a solution~$f^\star$ of~\eqref{eqn_funct_prob} of the form~\cite{scholkopf2001generalized, kimeldorf1971some}
\begin{equation}\label{eqn_representer}
	f^\star(\cdot) = \sum_{i = 1}^N a_i^\star k(\cdot,\bbx_i \,;\, w_0)
		\text{.}
\end{equation}
This celebrated variational result is known as the representer theorem and lies at the core of the success of RKHS methods by reducing the functional~\eqref{eqn_funct_prob} to the finite dimensional
\begin{prob}[$\text{PI}^\prime$]\label{eqn_finite_prob}
	\minimize_{\{a_i\} \in \reals}&
		&&\sum_{i= 1}^N \sum_{j = 1}^N a_i a_j k(\bbx_i,\bbx_j \,;\, w_0)
	\\
 	\subjectto& &&c(\yhat_i, y_i) \leq 0
 		\text{,} \quad i = 1, \dots, N
 		\text{,}
 	\\
 	&&&\yhat_i = f(\bbx_i) = \sum_{j = 1}^N a_j k(\bbx_i,\bbx_j \,;\, w_0)
 		\text{.}
\end{prob}

Their success notwithstanding, RKHS-based methods have limitations that can hinder their use in practice. Firstly, the kernel function~$k$~(i.e., the RKHS~$\ccalH_0$) must be chosen \emph{a priori} and given that it determines the shape of the solution~[see~\eqref{eqn_representer}], its choice directly affects the method performance. Even when the form of~$k$ is known, selecting the smoothness parameter~$w_0$ in~\eqref{eqn_orig_repres} can be challenging without application-specific prior knowledge. Secondly, functions RKHSs have limited capability to fit functions with heterogeneous degrees of smoothness, i.e., functions that vary slowly in some regions of their domains and rapidly in others~\cite{Donoho98m}. Although adapting~$k$ or the parameter~$w_0$ in~\eqref{eqn_representer} has been proposed to address this issue, \eqref{eqn_finite_prob} then becomes a non-convex program, foregoing efficient solutions and/or optimality guarantees. Finally, although the representer theorem allows us to use the finite dimensional~\eqref{eqn_finite_prob} to solve the functional~\eqref{eqn_funct_prob}, note that evaluating the solution~\eqref{eqn_representer} requires as many kernel evaluations as observations. Hence, the complexity of the representation grows with the sample size, which is infeasible for large data sets. Moreover, even if finding a sparse set of coefficients~$a_j$ in~\eqref{eqn_finite_prob} were tractable, which it is not~\cite{Natarajan95s, amaldi1998approximability}, the representation in~\eqref{eqn_representer} does not hold in the presence of a sparsity regularization. In fact, there often exist more parsimonious representations that fit the samples as well as any solution of~\eqref{eqn_finite_prob}~(see Remark~\ref{R:sparsity}).

To overcome these issues, the next section puts forward an integral parametrization of functions in RKHSs. This representation can be used to pose problems that can locally adapt not only the kernel centers, but also the kernel itself~(i.e., the RKHS). Moreover, it allows for regularizations beyond smoothness, most notably sparsity. Although the resulting optimization programs are non-convex and infinite dimensional, we show they have zero duality gap and can therefore be solved efficiently and exactly using conventional methods, such as~(stochastic) (sub)gradient descent. This also allows us to derive a new sparse representer theorem. % that precludes the smoothness prior.

\begin{figure}[tb]
  \includegraphics[width=0.4\textwidth]{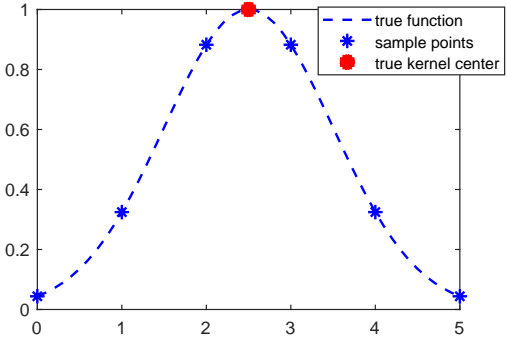}
  \caption{Illustration of Remark~\ref{R:sparsity} on the importance of kernel centers for model complexity.}
  \label{fig:resOneKernelex}
\end{figure}

\begin{remark}\label{R:sparsity}

When seeking parsimonious representations of functions in RKHSs, i.e., representations as in~\eqref{eqn_orig_repres} for which the~$a_j$ are sparse, the classical representer theorems leading to~\eqref{eqn_representer} do not apply. To see this is the case, take the counter-example illustrated in Figure~\ref{fig:resOneKernelex}. Let the sample points be taken from an underlying function composed of a single Gaussian kernel, namely
\begin{equation}
	y_i = \exp \left[ -\frac{(x_i - 2.5)^2}{2} \right]
		\text{,} \quad i = 1,\dots,N
		\text{,}
\end{equation}
where~$x_i \neq 2.5$ for all~$i$. What is more, assume that we know the correct RKHS~$\ccalH_0$, i.e., that the kernel function in~\eqref{eqn_orig_repres} and~\eqref{eqn_representer} is~$k(x,z) = \exp \left[ -(x - z)^2/2 \right]$. Then, it is clear that the most parsimonious function in~$\ccalH_0$ that fits the data is~$f^\prime(\cdot) = \exp \left[ -\frac{(\cdot - 2.5)^2}{2} \right]$. However, $f^\prime$ is not in the feasible set of~\eqref{eqn_finite_prob}. Hence, though it can find functions with the same approximation error, they will not be the simplest representation, as illustrated in Figure~\ref{fig:resOneKernelex}.

\end{remark}

%%%%%%%%%%%%%%%%%%%%%%%%%%%%%%%%%%%%%%%
%%% SECTION : Functional approach   %%%
%%%%%%%%%%%%%%%%%%%%%%%%%%%%%%%%%%%%%%%

\section{Learning an Integral Representation Instead} \label{sec:solution}

%!TEX root = mkl.tex
%%%%%%%%%%%%%%%%%%%%%%%%%%%%%%%%%%%%%%%
%%% SECTION : Functional approach   %%%
%%%%%%%%%%%%%%%%%%%%%%%%%%%%%%%%%%%%%%%

In the previous section, we have argued that classical RKHS-based methods suffer from three main drawbacks: (i)~the RKHS must be fixed \emph{a priori}, (ii)~they are suboptimal for functions with heterogeneous degrees of smoothness, and (iii)~the computational complexity of evaluating solutions is proportional to the sample size. Our first step towards addressing these issues is to extend~\eqref{eqn_orig_repres} by allowing different RKHSs at different centers. Throughout this work, we assume that the reproducing kernel family~$k$ is fixed~(e.g., Gaussian kernels) and adapt its smoothness parameter~(e.g., bandwidth). Nevertheless, all results hold for any parametrized dictionary of kernels, i.e., they need not belong to the same family. Explicitly, we write
\begin{equation}\label{eqn_general_rep}
	f(\bbx) = \lim_{n \to \infty} \sum_{j=1}^n a_j k(\bbx, \bbz_j \,;\, w_j)
		\text{.}
\end{equation}
The function~$f$ in~\eqref{eqn_general_rep} now lies in the sum space of a countable number of RKHSs~$\ccalH = \bigoplus_{j = 1}^\infty \ccalH_j$, where the~$\ccalH_j$ are defined by the parameters~$w_j$. Note that taking~$w_j = w_0$ for all~$j$ in~\eqref{eqn_general_rep} recovers the representation~\eqref{eqn_orig_repres} of a function in~$\ccalH_0$.

Similar formulations have been proposed to deal with the aforementioned issues, although optimally selecting~$w_j$~[issue~(i)--(ii)] and~$\bbz_j$~[issue~(iii)] remains an open problem~\cite{smola2000sparse, vincent2002kernel,tibshirani1996regression, fung2002minimal, jud2016sparse,gao2013sparse}. Our second step is therefore to tackle this hurdle by introducing an integral counterpart of~\eqref{eqn_general_rep}. Explicitly, we define the function
\begin{equation}\label{eqn_integral_rep}
	h(\cdot) = \int_{\ccalX \times \ccalW} \alpha(\bbz,w) k(\cdot,\bbz \,;\, w) d\bbz dw
		\text{,}
\end{equation}
where~$\ccalW$ is a compact subset of~$\reals$ and~$\alpha: \ccalX \times \ccalW \rightarrow \reals$ is in~$L_2(\ccalX \times \ccalW)$. Notice that, in contrast to~\eqref{eqn_orig_repres} and~\eqref{eqn_general_rep}, the expression in~\eqref{eqn_integral_rep} does not depend on a choice of centers or kernel parameters, thus addressing issues~(i) and~(ii). Before proceeding, we show that~\eqref{eqn_general_rep} and~\eqref{eqn_integral_rep} are essentially equivalent, i.e., that~\eqref{eqn_integral_rep} can essentially represent the functions in~$\ccalH$.

\begin{proposition}\label{T:equivalence}

Let~$k$ be a continuous reproducing kernel, i.e., $k(\cdot, \bbz ; \cdot)$ is continuous over the compact set~$\ccalX \times \ccalW$ for each~$\bbz \in \ccalX$. Then, for each~$f \in \ccalH_p$ there exists a sequence~$\{h_m\}$ of functions as in~\eqref{eqn_integral_rep} such that~$h_m \to f$ pointwise.

\end{proposition}

\begin{proof}

Consider the approximation of the identity~$r_m(x) = m \indicator\left[ \vert x \vert < 1/m \right]$ and note that~$r_m(x) \to \delta(x)$ weakly in the vague topology, i.e., $\int_\ccalD r_m(x) \varphi(x) \to \varphi(0)$ for all~$\varphi$ continuous and~$\ccalD$ compact~\cite{Rudin91f}. Now let~$f \in \ccalH_p$ be written as~$f(\cdot) = \sum_{j=1}^n a_j k(\cdot, \bbz_j \,;\, w_j)$ and take~$h_m(\cdot) = \int_{\ccalX \times \ccalW} \alpha_m(\bbz,w) k(\cdot,\bbz \,;\, w) d\bbz dw$ with
\begin{equation}\label{eqn_alpha_n}
	\alpha_m(\bbz,w) = \sum_{j = 1}^n a_j r_m(w - w_j) \prod_{k = 1}^p r_m([\bbz]_k-[\bbz_j]_k)
		\text{,}
\end{equation}
where~$[\bbz]_k$ indicates the~$k$-th element of the vector~$\bbz$. Note that~$\alpha_m \in L_2$, so that~$h_m$ is indeed of the form~\eqref{eqn_integral_rep}. Since the reproducing kernel is continuous, it readily holds from~\eqref{eqn_alpha_n} that~$g_m \to f$ pointwise for all~$\bbx \in \ccalX$.
\end{proof}

Proposition~\ref{T:equivalence} shows that there is no loss in using~\eqref{eqn_integral_rep} since it can essentially represent all functions in the sum space~$\ccalH$ of interest. Though this result is straightforward when~\eqref{eqn_integral_rep} defines the inner product in~$\ccalH$ or if~$\alpha$ allowed distributions~(Dirac deltas), the former is typically not the case~(except for sinc kernels) and the latter does not hold since~$\alpha \in L_2$. Also, it would appear that this integral representation did nothing but aggravate the computational complexity problems~[(iii)], now that the coefficients~$\alpha$ are functions. What is more, estimating~$\alpha$ is once again a functional problem. In the sequel, we address the computational complexity issue by explicitly minimizing the support of~$\alpha$ using sparse functional optimization programs~(SFPs). In Section~\ref{sec:dual}, we then show that these optimization problems can be solved efficiently and exactly using duality, providing an explicit algorithm to compute~$\alpha$.

\subsection{A Sparse Functional Formulation}

Proposition~\ref{T:equivalence} suggests that the centers and kernel parameters can be obtained by leveraging sparsity. Indeed, notice from~\eqref{eqn_alpha_n} that functions in~$\ccalH$ admit an integral representation~\eqref{eqn_integral_rep} in which~$\alpha$ is a superposition of bump functions centered around~$(\bbz_j,w_j)$, i.e., a function that vanishes over most of its domain. In other words, functions in~$\ccalH$ admit a \emph{sparse} integral parametrization. This observation motivates estimating~$f$, equivalently~$\alpha$, using the following SFP:
\begin{prob}\label{eqn_the_problem}
	\minimize_{\alpha \in L_2(\ccalX \times \ccalW)}&
		&&\frac{1}{2} \Vert \alpha \Vert_{L_2}^2 + \gamma \Vert \alpha \Vert_{L_0}
	\\
	\subjectto& &&c(\yhat_i, y_i) \leq 0 \text{,} \quad i = 1, \dots, N
		\text{,}
	\\
	&&&\yhat_i = f(\bbx_i) = \int_{\ccalX \times \ccalW}
		\alpha(\bbz, w) k(\bbx_i, \bbz \hspace{0.12em};\hspace{0.12em} w) d\bbz dw
		\text{,}
\end{prob} 
where~$\gamma \geq 0$ is a regularization factor that trades-off smoothness and sparsity; $\Vert \cdot \Vert_{L_2}$ denotes the~$L_2$-norm, which induces smoothness, enhances robustness to noise, and improves the numerical properties of the optimization problem; and~$\Vert \cdot \Vert_{L_0}$ denotes the~``$L_0$-norm,'' defined as the measure of the support of a function, i.e.,
\begin{equation}\label{eqn_L0}
	\Vert \alpha \Vert_{L_0} = \int_{\ccalX \times \ccalW}
		\indicator \left[ \alpha(\bbz, w) \neq 0 \right] d\bbz dw
		\text{,}
\end{equation}
where~$\indicator[x \neq 0] = 1$ if~$x \neq 0$ and zero otherwise. Notice that~\eqref{eqn_L0} is the functional counterpart of the discrete~``$\ell_0$-norm''.

Problem~\eqref{eqn_the_problem} seeks the function~$f \in \ccalH$ with the sparsest integral representation that fits the data according to the convex loss~$c$, which is both non-convex and infinite dimensional. As it is, it therefore appears to be intractable. Before addressing this matter, however, we argue that its solutions would indeed cope with the statistical and computational issues of classical RKHS methods. In Section~\ref{sec:dual}, we then derive efficient algorithms to obtain these solutions.

To be sure, \eqref{eqn_the_problem} precludes the choice of a specific RKHS~[issue~(i)] or kernel centers by leveraging the integral representation~\eqref{eqn_integral_rep}. Simultaneously, it enables the solution to locally adapt the RKHS over the domain to account for functions with heterogeneous degrees of smoothness~[issue~(ii)]. Finally, the sparse solutions of~\eqref{eqn_the_problem} can be used to obtain low complexity representations of functions in the sum space~$\ccalH$, i.e., representations with a small number of kernels~[issue~(iii)]. Intuitively, we can obtain a finite series from the integral representation~\eqref{eqn_integral_rep} by using the bumps in~$\alpha$ to determine the pair~$(\bbz_j,w_j)$ that define~\eqref{eqn_general_rep}, as suggested by~\eqref{eqn_alpha_n}. The~$a_j$ can then be obtained by directly minimizing~$c$. 

Although the remainder of this work studies the general~\eqref{eqn_the_problem}, two particular cases are of marked interest. First, the functional space of the solution is sometimes dictated by the application. For instance, one may seek bandlimited functions of a specific bandwidth. In this case, the reproducing kernel~$k$ and its parameter~$w_0$ are fixed and~\eqref{eqn_the_problem} reduces to
\begin{prob}[$\text{PII}^\prime$]\label{centersearch}
	\minimize_{\alpha \in L_2(\ccalX)}& &&\frac{1}{2} \Vert \alpha \Vert_{L_2}^2
		+ \gamma \Vert \alpha \Vert_{L_0}
	\\
	\subjectto& &&c(\yhat_i, y_i) \leq 0 \text{,} \quad i = 1, \dots, N
		\text{,}
	\\
	&&&\yhat_i = f(\bbx_i) = \int_{\ccalX}
		\alpha(\bbz) k(\bbx_i, \bbz \,;\, w_0) d\bbz
		\text{.}
\end{prob} 
Problem~\eqref{centersearch} sets to find solutions that use as few kernels as possible and its often tackled using greedy methods such as KOMP~\cite{vincent2002kernel}. Despite its success, its applications either implicitly relies on the representer theorem~\cite{koppel2017parsimonious}, which does not hold for sparse problems~(see Remark~\ref{R:sparsity}), or use a grid search over the space, which can quickly become prohibitive. In contrast, the solution of~\eqref{centersearch} is guaranteed to provide the sparsest integral representation. These can then be approximated using the aforementioned peak finding method to yield low complexity, discrete solutions. Problem~\eqref{centersearch} can therefore obtain solutions with the same cost and lower complexity, as illustrated in Section~\ref{sec:sims}.

Second, a set of candidate kernel centers might be available \emph{a priori}, given by domain experts or unsupervised learning techniques such as clustering. Problem~\eqref{eqn_w_only} can then be used to optimally select a subset of these centers as well as determine the appropriate RKHS. Explicitly,
\begin{prob}[$\text{PII}^{\prime\prime}$]\label{eqn_w_only}
	\minimize_{\alpha \in L_2(\ccalW)}&
		&&\sum_{j = 1}^M \left[ \frac{1}{2} \Vert \alpha_j \Vert_{L_2}^2
			+ \gamma \Vert \alpha_j \Vert_{L_0} \right]
	\\
	\subjectto& &&c(\yhat_i, y_i) \leq 0 \text{,} \quad i = 1, \dots, N
		\text{,}
	\\
	&&&\yhat_i = f(\bbx_i) = \sum_{j = 1}^M \int_{\ccalW}
		\alpha_j(w) k(\bbx_i, \bbz_j \,;\, w) dw
		\text{,}
\end{prob} 
where~$\bbz_j$, for~$j = 1,\dots,M$, are the predefined candidate centers. Observe that~\eqref{eqn_w_only} promotes the sparsity of each~$\alpha_j$, so that the coefficient of those centers that do not contribute to satisfy the fit constraint will vanish. Hence, the solution of~\eqref{eqn_w_only} effectively selects the smallest subset of candidate centers. Furthermore, by locally adapting the RKHS, it can further reduce the number of kernels~(centers) in the final solution by using less kernels to fit smoother portions of the data. The formulation~\eqref{eqn_w_only} is particularly attractive for high dimensional problems for which evaluating the integrals in~\eqref{eqn_the_problem} may be challenging.

In the next section, we derive a method for solving~\eqref{eqn_the_problem} exactly and efficiently. Solutions to problems~\eqref{centersearch} and~\eqref{eqn_w_only} were presented in~\cite{peifer2019sparse}~and~\cite{peifer2018locally} respectively. We do so by formulating the dual problem of~\eqref{eqn_the_problem} and showing that its solution can be used to obtain an optimal~$\alpha^\star$~(strong duality). This then allows us to propose efficient solutions for~\eqref{eqn_the_problem} by means of its dual problem. This strong duality result is also exploited to derive a new integral representer theorem~(Corollary~\ref{T:representer}) that accounts for sparsity.

%%%%%%%%%%%%%%%%%%%%%%%%%%%%%%%%%%%%%%%
%%% SECTION : Dual Problem          %%%
%%%%%%%%%%%%%%%%%%%%%%%%%%%%%%%%%%%%%%%
\section{Learning in the Dual Domain}\label{sec:dual}

%!TEX root = mkl.tex
%%%%%%%%%%%%%%%%%%%%%%%%%%%%%%%%%%%%%%%%%%%%%%%%%%%%%%%%%%%%%%%%%%%%%%%%%%%%%
%%%   S   E   C   T   I   O   N   %%%%%%%%%%%%%%%%%%%%%%%%%%%%%%%%%%%%%%%%%%%
%%%%%%%%%%%%%%%%%%%%%%%%%%%%%%%%%%%%%%%%%%%%%%%%%%%%%%%%%%%%%%%%%%%%%%%%%%%%%

Having argued that~\eqref{eqn_the_problem}~[or~\eqref{centersearch}--\eqref{eqn_w_only}] is worth solving, we now return to the issue of how. To understand the challenge of directly tackling~\eqref{eqn_the_problem}, observe that it is a non-convex, infinite dimensional optimization program. Moreover, its discrete version is in general NP-hard to solve~\cite{Natarajan95s}. In the sequel, we address these hurdles using duality. It is worth noting that duality is an established approach to solve semi-infinite convex programs~\cite{Shapiro06d, Tang13c, Candes14t}. Indeed, dual problems are convex and their dimension is equal to the number of constraints. Thus, they can be solved efficiently. Moreover, it is well-known that when the original problem is convex, its solutions can be obtained from solutions of the dual problem under mild conditions~\cite{Boyd04c}. Nevertheless, this is not the case of~\eqref{eqn_the_problem}.

In the sequel, we derive a method to solve~\eqref{eqn_the_problem} by first obtaining its dual problem in closed-form~(Section~\ref{sec_dual_problem}). Then, we show that strong duality holds~(Section~\ref{sec_strong}). Consequently, solutions of~\eqref{eqn_the_problem} can be obtained from solutions of its dual problem. We then conclude by showing how to efficiently solve the latter~(Section~\ref{sec_solver}).

\subsection{The dual problem of~\eqref{eqn_the_problem}}
	\label{sec_dual_problem}

To derive the dual problem of~\eqref{eqn_the_problem}, start by introducing the Lagrange multipliers~$\bblambda \in \reals^N$, associated with its equality constraints and~$\bbmu \in \reals^N_{+}$, associated with its inequality constraints. Its Lagrangian is then defined as
\begin{equation}\label{eqn_lagrangian}
\begin{aligned}
	\ccalL(\alpha,\boldsymbol{\yhat},\bblambda,\bbmu) &=
		\frac{1}{2} \Vert \alpha \Vert_{L_2}^2 + \gamma \Vert \alpha \Vert_{L_0}
    \\
	{}&- \sum_{i = 1}^N \bblambda_i \int
		\alpha(\bbz, w) k(\bbx_i, \bbz \,;\, w) d\bbz dw
	\\
	{}&+ \sum_{i = 1}^N \bblambda_i \yhat_i + \sum_{i = 1}^N \bbmu_i c(\yhat_i, y_i)
		\text{.}
\end{aligned} 
\end{equation}
For conciseness, we omit the set~$\ccalX \times \ccalW$ over which the integrals are computed. To proceed, obtain the dual function by minimizing the Lagrangian as in
\begin{equation}\label{eqn_dual_function} 
	g(\bblambda,\bbmu) = \min_{\substack{\alpha \in L_2 \\ \yhat_i \in \reals}}
		\ccalL(\alpha,\boldsymbol{\yhat},\bblambda,\bbmu)
		\text{.}
\end{equation}
The dual function is the objective of the dual problem, defined explicitly as
\begin{prob}[DII]\label{dual_prob} 
	\maximize_{\bblambda\in\reals^N,\, \bbmu \in\reals^N_+}& &&g(\bblambda,\bbmu)
		\text{.}
\end{prob}

Notice that the dual function~\eqref{eqn_dual_function} is concave regardless of the convexity of the original problem since it is the minimum over a set of affine functions in~$(\bblambda,\bbmu)$. What is more, despite the infinite dimensionality of the original problem, it is defined over~$2N$ variables. Hence, \eqref{dual_prob} is a finite dimensional convex program, which can be solved efficiently as long as its objective can be evaluated efficiently.

Yet, computing~$g$ involves a functional, non-convex problem. Indeed, notice that~\eqref{eqn_dual_function} can be separated as
\begin{equation}\label{eqn_dual_separated}
\begin{aligned}
	g(\bblambda,\bbmu) &= \min_{\yhat_i \in \reals} \ccalL_{\yhat}(\boldsymbol{\yhat},\bblambda,\bbmu)
		+  \min_{\alpha \in L_2} \ccalL_{\alpha}(\alpha,\bblambda)
		\text{,}
\end{aligned} 
\end{equation}
where
\begin{equation}\label{eqn_g_yhat}
	\ccalL_{\yhat}(\boldsymbol{\yhat},\bblambda,\bbmu) = 
		\sum_{i = 1}^N \bbmu_i c(\yhat_i, y_i) + \sum_{i = 1}^N \bblambda_i \yhat_i
\end{equation}
is the value a convex optimization problem, since~$c$ is convex and~$\bbmu_i \geq 0$, and
\begin{multline}\label{eqn_g_alpha}
	\ccalL_{\alpha}(\alpha, \bblambda) = \int \left[\vphantom{\sum}
		\frac{1}{2} \alpha^2(\bbz,w) + \gamma \indicator\left[ \alpha(\bbz,w) \neq 0 \right]
	\right.
	\\
	{}- \left.\vphantom{\sum}
		\sum_{i = 1}^N \bblambda_i \alpha(\bbz, w) k(\bbx_i, \bbz \,;\, w)
	\right] d\bbz dw
		\text{,}
\end{multline}
where we used the integral definitions of the~``$L_0$-norm'' in~\eqref{eqn_L0} and the~$L_2$-norm. Hence, evaluating~$g$ involves solving a non-convex functional optimization problem similar to the original~\eqref{eqn_the_problem}. Here, however, we can exploit separability to obtain the closed-form thresholding solution presented in the following proposition.

\begin{proposition}\label{prop_Lagrangian_solution}

A minimizer~$\alpha_d$ of~\eqref{eqn_g_alpha} is given by
\begin{equation}\label{alphaopt}  
	\alpha_d(\bbz,w ; \bblambda) = \begin{cases}
		\bar{\alpha}_d(\bbz,w ; \bblambda)
			\text{,} &\left\vert \bar{\alpha}_d(\bbz,w ; \bblambda) \right\vert > \sqrt{2\gamma}
	\\
		0 \text{,} &\text{otherwise}\\
	\end{cases}
\end{equation}
where~$\bar{\alpha}_d(\bbz,w ; \bblambda) = \sum_{i = 1}^N \bblambda_i k(\bbx_i, \bbz; w)$.

\end{proposition}

\begin{proof}

To obtain~\eqref{alphaopt}, we start by separating the objective of~\eqref{eqn_g_alpha} across~$\bbz$ and~$w$. To do so, we leverage the following lemma:

\begin{lemma}\label{T:minInt}
	Let~$F(\alpha,x)$ be a normal integrand, i.e., continuous in~$\alpha$ for fixed~$x$ and measurable in~$x$ for fixed~$\alpha$. Then,
	\begin{equation}\label{E:exchangeInf}
		\inf_{\alpha \in L_2} \int F(\alpha(x),x) dx =
			\int \inf_{\bar{\alpha} \in \reals} F(\bar{\alpha},x) dx
			\text{.}
	\end{equation}
\end{lemma}

\begin{proof}
See~\cite[Thm. 3A]{rockafellar1976integral}.
\end{proof}

\noindent Taking
\begin{equation}\label{eqn_F}
	F(\bar{\alpha},\bbz,w) = \frac{\bar{\alpha}^2}{2}
		+ \gamma \indicator\left[ \bar{\alpha} \neq 0 \right]
		- \sum_{i = 1}^N \bblambda_i k(\bbx_i, \bbz \,;\, w) \bar{\alpha}
		\text{,}
\end{equation}
in Lemma~\ref{T:minInt}, yields that minimizing~\eqref{eqn_g_alpha} is equivalent to minimizing~$F$ individually for each~$(\bbz,w)$. Although still non-convex, this problem is now scalar and admits a simple solution.

Indeed, notice that the indicator function in~\eqref{eqn_F} can only take two values depending on~$\bar{\alpha}$. Hence, its optimal value is the minimum of two cases: (i) if~$\bar{\alpha} = 0$, then~$F(0,\bbz,w) = 0$ for all~$(\bbz,w)$; alternatively, (ii)~if~$\bar{\alpha} \neq 0$, then~\eqref{eqn_F} becomes
\begin{equation}\label{eqn_Fprime}
	F^\prime(\bar{\alpha},\bbz,w) = \frac{\bar{\alpha}^2}{2}
		- \sum_{i = 1}^N \bblambda_i k(\bbx_i, \bbz \,;\, w) \bar{\alpha} + \gamma
		\text{,}
\end{equation}
whose minimization is a quadratic problem with closed-form solution
\begin{equation}\label{eqn_bar_alpha_star}
	\bar{\alpha}^\star(\bbz,w) = \argmin_{\bar{\alpha} \in \reals} F^\prime(\bar{\alpha},\bbz,w)
		= \sum_{i = 1}^N \bblambda_i k(\bbx_i, \bbz \,;\, w)
		\text{,}
\end{equation}
so that~$\min_{\bar{\alpha} \neq 0} F(\bar{\alpha}^\star,\bbz,w) = \gamma -  \bar{\alpha}^\star(\bbz,w)^2/2$. Immediately, $\alpha_d(\bbz,w) = \bar{\alpha}^\star(\bbz,w)$ if~$\gamma -  \bar{\alpha}_d(\bbz,w)^2/2 < 0$ or~$\alpha(\bbz,w)$ vanishes, which yields~\eqref{alphaopt}.
\end{proof}

Hence, despite the non-convexity and infinite dimensionality nature of~\eqref{eqn_g_alpha}, it admits an explicit solution in the form of the thresholded function~\eqref{alphaopt}. We can thus evaluate the convex objective of~\eqref{dual_prob}, that can then be solved using classical convex optimization tools such as~(stochastic) (sub)gradient ascent~(see Section~\ref{sec_solver}). Still, the question remains of whether this is a worthwhile endeavor. Indeed, \eqref{eqn_dual_function} can be interpreted as a relaxation of the hard constraints in~\eqref{eqn_the_problem}, so that~\eqref{dual_prob} provides a lower bound on its optimal value. Though we can evaluate~\eqref{alphaopt} on the solution~$(\bblambda^\star, \bbmu^\star)$ of~\eqref{dual_prob}, there are no guarantees that it is a solution of~\eqref{eqn_the_problem}. We tackle this challenge in the sequel.

\subsection{Strong duality and the integral representer theorem}
	\label{sec_strong}

Due to the non-convexity of the original problem, the only immediate guarantee we can give about the optimal value of~\eqref{dual_prob} is that it is a lower bound on the optimal value of~\eqref{eqn_the_problem}. The central technical result of this section shows that strong duality holds for~\eqref{eqn_the_problem}, i.e., it has null duality gap~(Theorem~\ref{thrm_zerodual}). This result has deep implications for the problems we posed in Section~\ref{sec:problem}. First, it implies that we can obtain a solution of~\eqref{eqn_the_problem} by solving its~\eqref{dual_prob}~(Corollary~\ref{T:primal_recovery}). Second, it allows us to write a representer theorem similar to the original ones stating that the solution~$\alpha^\star$ of~\eqref{eqn_the_problem} is a linear combination of kernels evaluated at the data points.

Let us start with the main theorem:

\begin{theorem}\label{thrm_zerodual}

Strong duality holds for~\eqref{eqn_the_problem} if the kernel~$k(\cdot,\bbz \,;\, \omega)$ has no point masses and Slater's condition is met. In other words, if~$P$ is the optimal solution of~\eqref{eqn_the_problem} and~$D$ is the optimal solution of~\eqref{dual_prob}, then the duality gap~$P-D$ vanishes.

\end{theorem} 

\begin{proof}

This result can be found in~\cite{Chamon18s}. For ease of reference, a proof is provided in Appendix~\ref{A:zerodual}.
\end{proof}

Theorem~\eqref{thrm_zerodual} states that, though it is not a convex program,~\eqref{eqn_the_problem} has null duality gap. Immediately, we obtain the following corollary:

\begin{corollary}\label{T:primal_recovery}

Let~$(\bblambda^\star,\bbmu^\star)$ be a solution of~\eqref{dual_prob} and assume~$k \in L_2$ and analytic. Then, $\alpha_d^\star(\cdot,\cdot) = \alpha_d(\cdot,\cdot ; \bblambda^\star)$ is a solution of~\eqref{eqn_the_problem} for~$\alpha_d$ as in~\eqref{alphaopt}.

\end{corollary}

\begin{proof}
See Appendix~\ref{A:primal_recovery}.
\end{proof}

\noindent Thus, despite the apparent challenges, \eqref{eqn_the_problem} is tractable and can be solved efficiently and exactly using duality~(as detailed in Section~\ref{sec_solver}). It is worth noting that the technical hypotheses of Corollary~\ref{T:primal_recovery} are mild and hold for all commonly used reproducing kernels, such as Gaussian, polynomial, and sinc kernels. A fundamental feature of this approach is that it solves the functional problem without relying on discretizations. This is of utmost importance since discretizing~\eqref{eqn_the_problem} can lead to NP-hard, large dimensional, and potentially ill-conditioned problems. Moreover, it tackles the sparse problem directly instead of using convex relaxations.

Another fundamental implication of Theorem~\ref{thrm_zerodual} is the following integral representer theorem.

\begin{corollary}[Integral representer theorem]\label{T:representer}

A solution~$\alpha^\star$ of~\eqref{eqn_the_problem} can be obtained by thresholding a parametrized family of functions~$\bar{\alpha}^\star_w \in \ccalH_w$, where~$\ccalH_w$ is the RKHS induced by the kernel~$k(\cdot,\cdot;w)$. In fact, $\bar{\alpha}^\star_w$ lives in a finite dimensional subspace of~$\ccalH_w$ spanned by the kernels evaluated at the data points. Explicitly, there exist~$a_i \in \reals$ such that
\begin{equation}\label{eqn_alpha_bar}
	\alpha^\star(\cdot,w) = \bar{\alpha}^\star_w(\cdot) = \sum_{i = 1}^N a_i k(\bbx_i, \cdot; w)
		\text{.}
\end{equation}
\end{corollary}

\begin{proof}

From Corollary~\ref{T:primal_recovery}, $\alpha^\star = \alpha_d^\star$ almost everywhere with~$\alpha_d^\star(\bbz,w) = \alpha_d(\bbz,w;\bblambda_i^\star)$. Thus, the corollary stems from~\eqref{alphaopt} for~$\alpha_w(\cdot) = \sum_{i = 1}^N \bblambda_i^\star k(\bbx_i, \cdot; w)$, so that~$\alpha_w \in \ccalH_w$ by definition and~$\alpha^\star(\bbz,w) = \alpha_w(\bbz) \indicator(\vert \alpha_w(\bbz) \vert > \sqrt{2\gamma})$.
\end{proof}

Recall that the classical representer theorem~\cite{kimeldorf1971some, scholkopf2001generalized} reduces the functional problem~\eqref{eqn_funct_prob} to the finite dimensional problem~\eqref{eqn_finite_prob} by showing that it has a solution of the form~\eqref{eqn_representer} that lies in the span of the kernels evaluated at the data points. Likewise, Theorem~\ref{T:representer} states that a solution of~\eqref{eqn_the_problem} can be obtained from a family of functions with the similar finite representation~\eqref{eqn_alpha_bar}. Note, however, that although classical representer theorems do not account for sparse solution~(as argued in Remark~\ref{R:sparsity}), Theorem~\ref{T:representer} holds in the presence of the ``$L_0$-norm'' regularizer. The cost of doing so is adding a layer of indirection with respect to the original functional problem: whereas the classical representer theorem yields a function in the RKHS directly from the parameter~$a_i$ through~\eqref{eqn_representer}, Corollary~\ref{T:representer} states that the functional parameters of the integral representation~\eqref{eqn_integral_rep} of this function can be obtained from the~$a_i$. Still, the problem it solves is considerably more complex and comprehensive than~\eqref{eqn_funct_prob}.

It is worth noting that although~\eqref{eqn_the_problem} searches for~$\alpha \in L_2$, Corollary~\ref{T:representer} shows its solution depends only on a family of functions belonging to the RKHSs considered. Explicitly, the solution~$\alpha^\star(\cdot,w)$ of~\eqref{eqn_the_problem} is a thresholded version of a function in the RKHS~$\ccalH_w$ with reproducing kernel~$k(\cdot,\cdot;w)$. In particular, when considering the problem without sparsity~($\gamma = 0$), we have that~$\alpha^\star(\cdot,w) \in \ccalH_w$. In the case of~\eqref{centersearch}, this further simplifies to~$\alpha^\star \in \ccalH_0$.

This observation allows us to think of~\eqref{eqn_integral_rep} as building the function~$h^\star$ point-by-point by integrating the value of partial $L_2$-inner products between the reproducing kernel of~$\ccalH_w$ and a function~$\bar{\alpha}_w \in \ccalH_w$. To be more specific, notice that~\eqref{eqn_integral_rep} can be written as the iterated integrals
\begin{equation}\label{eqn_iterated}
	h^\star(\cdot) = \int_{\overline{\ccalW}}
		\left[ \int_{\overline{\ccalX}} \bar{\alpha}_w(\bbz) k(\cdot,\bbz ; w) d\bbz \right] dw
		\text{,}
\end{equation}
where~$\overline{\ccalX} \subseteq \ccalX$, $\overline{\ccalW} \subseteq \ccalW$, and~$\overline{\ccalX} \times \overline{\ccalW}$ is the set induced by the support of~$\alpha^\star$, i.e., $\{ (\bbz,w) \in \ccalX \times \ccalW \mid \vert\alpha(\bbz,w)\vert > \sqrt{2\gamma} \}$. The innermost integral in~\eqref{eqn_iterated} can be interpreted as an inner product in~$L_2$ between~$\bar{\alpha}_w$ and~$k(\cdot,\bbz ; w)$ computed only where the magnitude of~$\bar{\alpha}_w$ is large enough, defined by the regularization parameter~$\gamma$. This sort of trimmed inner product is linked to robust projections found in different statistical methods~\cite{Chen13r, Feng14r}. The outer integral then accumulates the projections of~$\bar{\alpha}_w$ over the relevant subset~$\overline{\ccalW}$ of RKHSs considered to form the functional solution.

Before proceeding, it is worth noting that Theorem~\ref{thrm_zerodual} holds under very mild conditions. Indeed, the reproducing kernels typically used in applications, such as polynomial or Gaussian kernels, do not have point masses. In fact, if the function of interest is in~$L_2$, i.e., $\ccalH \subseteq L_2$, then we need not consider kernels containing Dirac deltas since they are not square integrable. As for Slater's condition~\cite{Boyd04c}, the infinite dimensionality of~$\alpha$ makes it so we can always find one that perfectly interpolates the data, though it may neither be smooth nor have a sparse representation. Hence, finding a strictly feasible solution of~\eqref{eqn_the_problem} is straightforward for most~$c$.

In the next section, we leverage the closed-form of the dual function from Proposition~\ref{prop_Lagrangian_solution} and the strong duality result from Theorem~\ref{thrm_zerodual} to obtain an explicit algorithm to solve~\eqref{eqn_the_problem}.

\subsection{Dual gradient ascent}
	\label{sec_solver}

\begin{algorithm}[t]
\caption{Stochastic optimization for~\eqref{eqn_the_problem}}\label{alg_solver}
\begin{algorithmic}[1]
	\State Initialize~$\bblambda_i(0)$ and~$\bbmu_i(0) > 0$ 
	
	\For {$t = 0, 1, \dots, T$}
		\State Evaluate the supergradient~$d_{\bbmu_i}(t) = c\left( \yhat_{d,i}(t), y_i \right)$ for
		\begin{equation*}
			\yhat_{d,i}(t) = \argmin_{\yhat_i} \sum_{i = 1}^N \bbmu_i(t) c(\yhat_i, y_i)
				- \sum_{i = 1}^N \bblambda_i(t) \yhat_i
		\end{equation*}

		\State Draw~$\{(\bbz_k,w_k)\}$, $k = 1,\dots,B$, uniformly at random and compute the stochastic supergradient
		\begin{equation*}
			\hat{d}_{\bblambda_i}(t) = \yhat_{d,i}(t)
				- \frac{1}{B} \sum_{k = 1}^B \alpha_d(\bbz_k, w_k; \bblambda_i(t)) k(\bbx_i, \bbz_k; w_k)
		\end{equation*}

		\State Update the dual variables:
		\begin{align*}
			\bblambda_i(t+1) &= \bblambda_i(t) + \eta_\lambda \hat{d}_{\bblambda_i}(t)
			\\
			\bbmu_i(t+1) &= \left[ \bbmu_i(t) + \eta_\mu d_{\bbmu_i}(t) \right]_+
		\end{align*}

	\EndFor 

	\State Evaluate the primal solution as
	\begin{equation*}
		\alpha^\star(\bbz, w) = \begin{cases}
			\bar{\alpha}^\star(\bbz, w)
			\text{,} &\left\vert \bar{\alpha}^\star(\bbz,w) \right\vert > \sqrt{2\gamma}
			\\
			0 \text{,} &\text{otherwise}
		\end{cases}
	\end{equation*}
	for~$\bar{\alpha}^\star(\bbz, w) = \sum_{i = 1}^N \bblambda_i(T) k(\bbx_i, \bbz; w)$
\end{algorithmic}

\end{algorithm}

Theorem~\ref{thrm_zerodual} shows that we can obtain a solution of~\eqref{eqn_the_problem} through~\eqref{alphaopt}. Still, although~\eqref{alphaopt} can be evaluated using the closed form expression from Proposition~\ref{prop_Lagrangian_solution}, it requires the optimal dual variables~$(\bblambda^\star,\bbmu^\star)$. In this section, we propose a projected supergradient ascent method to solve~\eqref{dual_prob}~(Algorithm~\ref{alg_solver}). Alternatively, other standard convex optimization algorithms can be used to exploit structure in the solution of~\eqref{dual_prob}. For instance, efficient solvers based on coordinate ascent can be leveraged to solve large-scale instances~\cite{Bertsekas15c}.

Start by recalling that a supergradient of a function~$f: \ccalD \to \reals$ at~$\bbx \in \ccalD \subseteq \reals^n$ is any vector~$\bbd$ such that~$f(\bby) \leq f(\bbx) + \bbd^T (\bby - \bbx)$ for all~$\bby \in \ccalD$. Though supergradients may not be an ascent direction at~$\bbx$, taking small steps in its direction decreases the distance to any maximizer of a convex function~$f$~\cite{Boyd04c}. Thus, we can solve~\eqref{dual_prob} by repeating, for~$t = 0,1,\dots$,
\begin{subequations}\label{eqn_grad_ascent}
\begin{align}
	\bblambda_i(t+1) &= \bblambda_i(t) + \eta_\lambda d_{\bblambda_i}(\bblambda_i,\bbmu_i)
		\text{,}
		\label{eqn_grad_ascent_lambda}
	\\
	\bbmu_i(t+1) &= \left[ \bbmu_i(t) + \eta_\mu d_{\bbmu_i}(\bblambda_i,\bbmu_i) \right]_+
		\text{,}
		\label{eqn_grad_ascent_mu}
\end{align}
\end{subequations}
where~$\bbd_{\bblambda}^{(t)},\bbd_{\bbmu}^{(t)}$ are the supergradients of~$\bblambda$ and~$\bbmu$, respectively, $\eta_{\lambda}, \eta_{\mu}>0$ are step sizes, and~$[x]_+ = \max(0,x)$. The projection of~$\bbmu_i$ on the non-negative numbers guarantees that the constraints of~\eqref{dual_prob} are satisfied. The supergradient in~\eqref{eqn_grad_ascent} are readily obtained from the constraint violation of the dual minimizers~\cite{Boyd04c}. Explicitly, let~$\yhat_{d,i}(\bblambda,\bbmu)$ and~$\alpha_d(\bbz, w;\bblambda)$ be minimizers of~\eqref{eqn_g_yhat} and~\eqref{eqn_g_alpha} respectively. Then,
\begin{subequations}\label{eqn_subgrad}
\begin{align}
	d_{\bblambda_i}(\bblambda,\bbmu) &= \yhat_{d,i}(\bblambda, \bbmu)
		- \int \alpha_d(\bbz, w; \bblambda) k(\bbx_i, \bbz; w) d\bbz dw
		\text{,}
		\label{eqn_subgrad_lambda}
	\\
	d_{\bbmu_i}(\bblambda,\bbmu) &= c\left( \yhat_{d,i}(\bblambda, \bbmu), y_i \right)
		\label{eqn_subgrad_mu}
		\text{.}
\end{align}
\end{subequations}

Since~$\yhat_{d,i}$ is the solution of the convex optimization problem~\eqref{eqn_g_yhat}, the update for the dual variables~$\bbmu_i$ in~\eqref{eqn_grad_ascent_mu} can be efficiently evaluated using~\eqref{eqn_subgrad_mu}. The update expression for~$\bblambda_i$ in~\eqref{eqn_grad_ascent_lambda}, however, requires that the integral in~\eqref{eqn_subgrad_lambda} be evaluated. To do so, we can either use numerical integration methods, since~$\alpha_d$ is available in closed-form from~\eqref{alphaopt}, or rely on Monte Carlo methods. The latter approach is especially interesting because it can be integrated with the optimization iterations in~\eqref{eqn_grad_ascent} to obtain a stochastic supergradient ascent algorithm summarized in Algorithm~\ref{alg_solver}. Since Monte Carlo gives an unbiased estimate of~$d_{\bblambda_i}$, typical convergence guarantees for stochastic optimization apply~\cite{Ruszczynski86c, Ribeiro10e, Bottou16o}.

%%%%%%%%%%%%%%%%%%%%%%%%%%%%%%%
%%% SECTION : Simulations   %%%
%%%%%%%%%%%%%%%%%%%%%%%%%%%%%%%

\section{Numerical Experiments} \label{sec:sims}

%!TEX root = mkl.tex
%%%%%%%%%%%%%%%%%%%%%%%%%%%%%%%
%%% SECTION : Simulations   %%%
%%%%%%%%%%%%%%%%%%%%%%%%%%%%%%%

In the previous sections we have claimed that our algorithm can estimate (i) kernel widths \eqref{eqn_w_only}, (ii) kernel centers \eqref{centersearch}, and (iii) kernels of varying centers and widths \eqref{eqn_the_problem}. In this section we show through a sample signal, how we can achieve claim (i). Then we show how moving from (i) to (iii) reduces complexity.
In our discussion about the complexity of the representation in section \ref{complexity}, we show how we can achieve (ii) on random signals of fixed width. In section \ref{var_smooth}, we solve \eqref{eqn_the_problem} for a signal of varying degrees of smoothness and show how we can reduce complexity regardless of sample size. Lastly, in sections \ref{Localization} and \ref{MNIST} we apply our algorithm to solve \eqref{eqn_the_problem} and \eqref{eqn_w_only} on two examples of real applications: a user localization problem and a digit classification problem.

For the estimation, we search over functions in the family of RKHSs, which have Gaussian functions as kernels

\begin{equation}
k(\bbx,\bbx^\prime) = exp \left\lbrace \frac{-\Vert \bbx - \bbx^\prime \Vert^2 }{2 w^2} \right\rbrace,
\end{equation}
where width of the kernel is directly proportional to the hyper-parameter $w$.  

To start, the effect of the choice of RKHS on the performance of a learning algorithm is examined. To this end, a signal, which lies in the RKHS with a Gaussian kernel of width $w_0 = 0.453$ is constructed. The classical problem in \eqref{eqn_finite_prob} is compared to the problem presented in \eqref{eqn_w_only}. A grid search is used to examine the performance of \eqref{eqn_finite_prob} for different values of $w$. The value of $w_0$ was chosen such that it would not be directly on the grid, since in practice it is unlikely to include the value of the width of the originating signal. 
We generate $S$ signals of the form

\begin{equation}\label{eq:mixedGauss}
f_{j}(x) = \sum_{i=1}^{m} a_i \times \exp\left[-\frac{\Vert\bbx-\tilde{\bbx}_i\Vert^2}{2*w_0^2}\right] + \xi_j
\end{equation}
with $j=1\ldots S$. For each $f_j$ a training set of $N = 50$ samples was generated with  $m=10$.
The amplitude $a_i$ of each function is selected at random from a uniform distribution $\mathcal{U}(1,2)$. The $\tilde{\bbx}_i$ are i.i.d random variables drawn from the uniform distribution $\mathcal{U}(1,2)$ and the $\xi_j$  are i.i.d. random variables drawn from $\ccalN(0, 10^{-3})$, which represent the noise. 

It should be noted that, given sufficient iterations, well chosen step sizes, and a large $\gamma$, our method can approximate point masses. However, smoother approximations of the point masses can be obtained by using only few iterations. Additionally, these smooth approximations are more robust to the choice of the tuning parameters. Kernel centers and widths can subsequently be obtained by selecting the extreme points of the function $\alpha(\bbz, w)$, since the optimal $\alpha(\bbz,w)$ is a function of $w$, the kernel  width, and $\bbz$ the kernel centers. Kernels using the widths and centers approximated from the extreme points of the $\alpha(\bbz,w)$ are used to train a least squared estimator.

A grid search is performed for problem \eqref{eqn_finite_prob} by uniformly sampling $w$ over the interval [0,1] at 0.1 increments. The problem \eqref{eqn_w_only} is solved using $\gamma=4000$, $\eta_\lambda = 0.001$, $\eta_\mu = 0.1$ and $T=5000$. The performance of the two algorithms is compared over $1000$ repetitions of the sampled signal, each with a training set of size $N = 100$ and a test set of size $N_{test}=1000$. The MSE of \eqref{eqn_finite_prob} decreases as the value of $w$ increases---see Figure \ref{fig:GridSearch}. Due to the non-uniform sampling of the signal, smoother kernels on average have a better performance. In areas, in which the sampling is sparse, the thinner kernels cannot represent the signal between the samples. Additionally, the thinner kernels are more likely to overfit to the noise than the smoother kernels. However, the smoother kernels cannot model the faster variation in the signal well. In contrast \eqref{eqn_w_only} finds a sparse solution, which uses 14 kernels on average, of varying smoothness, with an average MSE of $0.0457$, which can both take advantage of the ability of smoother kernels to avoid overfitting as well as thinner kernels to model fast variation. Indeed, we observe in Figure \ref{fig:Histogramw} that our algorithm chooses a mixture of kernels of width around $0.453$ and kernels of width $1$. 
\begin{figure}[tb!]
  \includegraphics[width=0.44\textwidth]{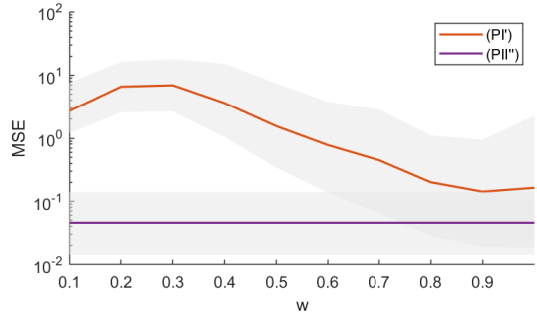}
  \caption{MSE obtained by \eqref{eqn_w_only} and \eqref{eqn_finite_prob} over 1000 repetitions of random sampling of the signal in \eqref{eq:mixedGauss}. \eqref{eqn_finite_prob} is solved over different values of $w$ over a grid on the interval $[0.1,1]$. \eqref{eqn_w_only} finds the width as part of the algorithm and is presented for comparison with \eqref{eqn_finite_prob}. The standard deviation around each mean is plotted in gray for both \eqref{eqn_w_only} and \eqref{eqn_finite_prob}. The figure shows that the selection of the width within the algorithm gives the advantage of a lower mean generalization error.}
  \label{fig:GridSearch}
\end{figure}

\begin{figure}[tb!]
  \includegraphics[width=0.44\textwidth]{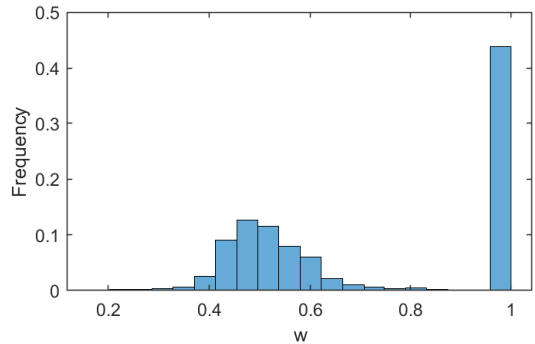}
  \caption{Histogram of the widths found using \eqref{eqn_w_only} over 1000 repetitions of random sampling of the signal in \eqref{eq:mixedGauss}. On average, $14$ kernels were selected for the representation of the function out of which an average of $6$ kernels have a width of $1$.}
  \label{fig:Histogramw}
\end{figure}

Smoother kernels perform better because of the random sampling combined with the restriction of only using kernels centered at the sample points. 
Therefore, we investigate the effect of solving problem \eqref{eqn_the_problem} which finds both kernel centers and kernel widths.
Problem \eqref{eqn_the_problem} is solved using $\gamma=1000$, $\eta_\lambda = 0.01$, $\eta_\mu = 1$ and $T=1000$ over 1000 randomly sampled training sets, and results in an MSE of $0.0588$. Although the MSE of \eqref{eqn_the_problem} is similar to that of \eqref{eqn_w_only}, it is important to note that by placing kernels arbitrarily we are able to better estimate the width of the kernel: by comparing Figure \ref{fig:HistogramMoving} to Figure \ref{fig:Histogramw} it can be seen that \eqref{eqn_the_problem} uses only 1 to 2 kernels per representation of width 1 whereas \eqref{eqn_w_only} uses on average 6 kernels of width 1. Moreover, we consistently obtain representations of lower complexity when solving \eqref{eqn_the_problem}---see Figure \ref{fig:NumKernel}.

\begin{figure}[tb!]
  \includegraphics[width=0.44\textwidth]{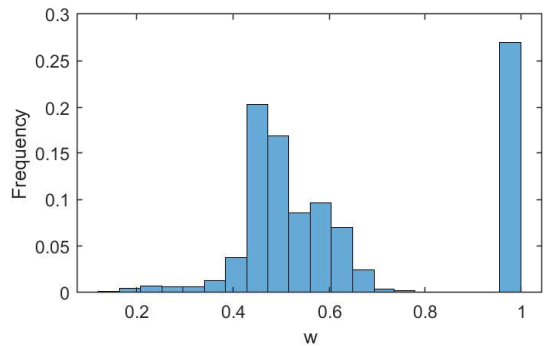}
  \caption{Histogram of the widths found using \eqref{eqn_the_problem} over 1000 repetitions of random sampling of the signal in \eqref{eq:mixedGauss}. On average a representation had 6 kernels out of which between 1 and 2 kernels had a width of $w=1$ and 4 kernels had a width in the interval $[0.384,0.648]$.}
  \label{fig:HistogramMoving}
\end{figure}
\begin{figure}[tb!]
  \includegraphics[width=0.44\textwidth]{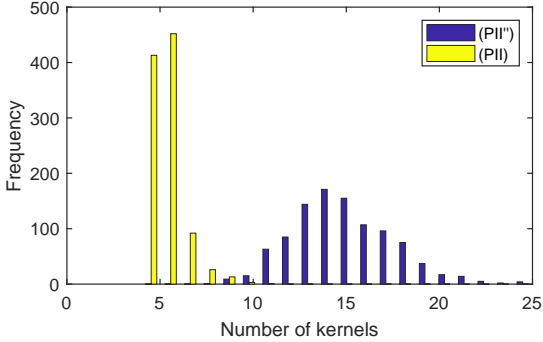}
  \caption{Histogram of the number of kernels in the representation of the estimated functions by solving problems \eqref{eqn_w_only} and \eqref{eqn_the_problem}. \eqref{eqn_the_problem} achieves a lower complexity representation by moving the centers in addition to the widths.}
  \label{fig:NumKernel}
\end{figure}

\subsection{Examining the Complexity of the Solution}\label{complexity}

So far we have shown that the complexity of the formulation can be reduced by moving centers in addition to moving the width. To further explore the effect of kernel centers on the complexity of the solution, we compare the performance of \eqref{centersearch} to that of kernel orthogonal matching pursuit (KOMP) with pre-fitting (see \cite{vincent2002kernel, koppel2017parsimonious}), for a simulated signal as in \eqref{eq:mixedGauss}. KOMP takes an initial function and a set of sample points and tries to estimate it by a parsimonious function of a lower complexity. As a backwards feature selection metheod, the algorithm starts by including all samples and then reduces the complexity of the function by reducing one feature at a time. The KOMP algorithm in \cite{vincent2002kernel, koppel2017parsimonious} was modified by changing the stopping criteria to be the estimation error, rather than the distance to the original function. This stopping criteria allows us to compare the sparsity needed to obtain similar estimation error.

The signal was sampled from the function in \eqref{eq:mixedGauss}~using $w_0=0.5$ and $m = [ 5, 10, 20]$ by generating $N = [2m, 4m, 6m]$ samples for each function, thus creating $9$ different sample size and signal pairs. The problem in \eqref{centersearch} was solved using $\gamma=30$, $\eta_\lambda = 0.05$, $\eta_\mu = 0.1$ and $T=1000$. Subsequently, a least squares algorithm was trained using kernels at the location found by our algorithm. Both our method and KOMP used $w = 0.5$ as the kernel hyper-parameter.

The number of kernels needed to obtain the same MSE is compared over $1000$ realizations of each signal between \eqref{centersearch} and KOMP. When the number of samples is at least $30$, our method is able to find a sparser representation  $100\%$ of the times. In the cases with fewer samples the problem is likely undersampled, such that the estimation of the function is more difficult. Figure \ref{KompNum} shows two cases in which $20$ samples are simulated, where $m=5$ and $m=10$. In both cases, our method finds sparser representations in $99\%$ of the realizations. 
When $10$ kernels and $5$ kernels are superimposed,  our method finds a representation which is less sparse in only $0.4\%$  and $0.3\%$ of realization respectively. Lastly, when the signal is a weighted sum of 5 functions and only 10 samples are generated, our method cannot find a sparser solution for $3.7\%$ of the realizations. 
\begin{figure*}[tb!]
\begin{subfigure}{.45\textwidth}
  \centering
  \includegraphics[width=.9\linewidth]{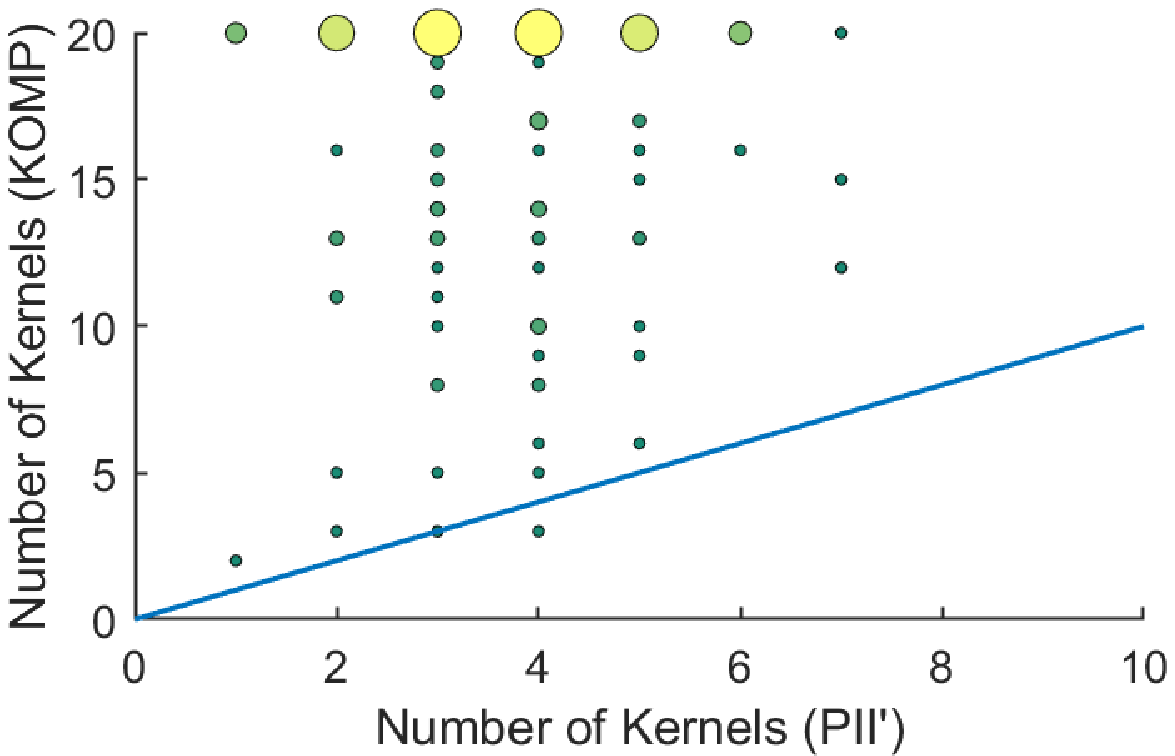}
  \caption{}
  \label{fig:sfig1}
\end{subfigure}%
\begin{subfigure}{.45\textwidth}
  \centering
  \includegraphics[width=.9\linewidth]{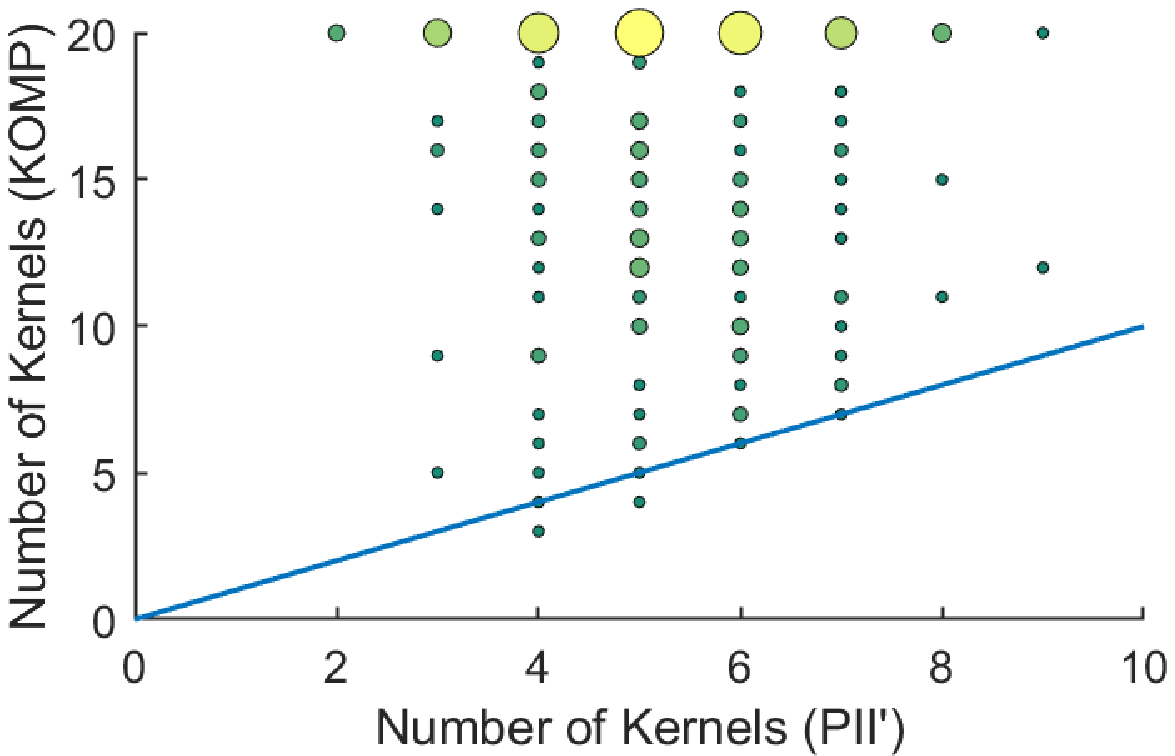}
  \caption{}
  \label{fig:sfig2}
\end{subfigure}
\caption{Comparison of the complexity of the representation of \eqref{centersearch} and KOMP for a similar MSE over 1000 realizations. In Figure (a) 5 Gaussian functions were used to simulate the signal. In Figure (b) 10 Gaussian functions were used to simulate the signal. In both cases \eqref{centersearch} achieves a lower complexity for $99\%$ of the realizations.
}\label{KompNum}
\end{figure*}

The generalization MSE was compared between the two methods for different levels of sparsity. Figure \ref{fig:FixedMSE} shows the changes in generalization MSE  as the number of kernels used in the representation increases. $1000$ realizations of a signal with m = $10$ and a training set of size N = $100$ were used. The ability of our method to place kernels at any location, beyond the training set, allows it to achieve significantly lower errors compared to KOMP at any sparsity level. As the number of kernels used increases, the difference in performance between the two methods decreases. At approximately 25 kernels the performance of our method plateaus. Comparatively, KOMP achieves a plateau when the representation holds 50 kernels. 
\begin{figure}[tb!]
\centering
  \includegraphics[width=.9\linewidth]{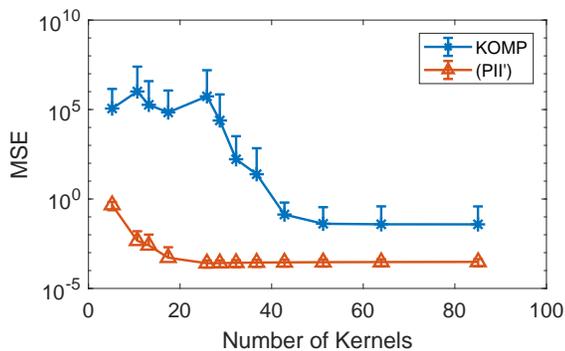}
  \caption{Generalization MSE as a function of number of kernels for KOMP and \eqref{centersearch} over 1000 realizations of the signal in \eqref{eq:mixedGauss}.}\label{fig:FixedMSE}
\end{figure}

\subsection{Varying Degrees of Smoothness} \label{var_smooth}
In the previous sections we have only considered signals from functions belonging to an RKHS in the family of RKHSs with Gaussian kernels. In this section we explore the effect of sample size on the complexity of the representation and the MSE on a signal of varying degrees of smoothness.
To this end, a signal of varying smoothness is simulated using the following equation:
\begin{equation}\label{sintsquared}
y_i = \sin(0.5 \pi x_i^2) + \xi_i
\end{equation}
where $\xi_i \in \mathcal{N}(0, 10^{-3})$ represents the noise. 

The solution of problem \eqref{eqn_the_problem} was compared to destructive KOMP, with the stopping criteria set to be the desired number of kernels rather than the distance from the original function. This stopping criteria allows us to have a fair comparison between our method and KOMP by using equally sparse functions. The problem in \eqref{eqn_the_problem} was solved using $\gamma=2$, $\eta_\lambda = 0.001$, $\eta_\mu = 30$ and $T=1000$.
Sample sizes of 51, 101, 201, 301, 401, and 501 were created by uniformly sampling in the interval $[-5,5]$. Test sets of 1000 samples randomly selected on the interval $[-5,5]$ were created. Using the method of selecting kernel centers and widths by selecting the peaks of the function $\alpha(\bbz,w)$, our method finds a representation with 26 kernels regardless of the sample size. 
It can be seen in Figure \ref{fig:MSEvSample} that in addition to the number of kernels being consistent across all sample sizes the MSE is also consistent for our method. The MSE of the estimation using KOMP, however, increases as the sample size grows. 
\begin{figure}[tb!]
\centering
 \includegraphics[width=0.45\textwidth]{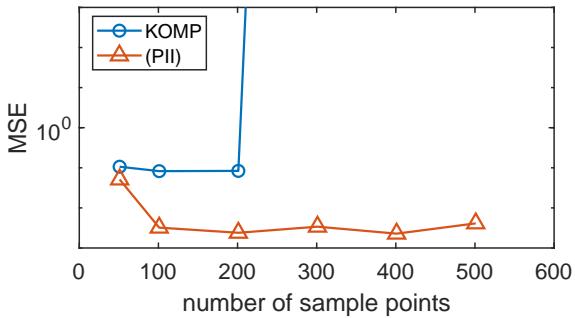}
 \caption{MSE for varying sample sizes using \eqref{eqn_the_problem} and  KOMP with 26 kernels over 100 realizations of the signal in \eqref{sintsquared}.}
 \label{fig:MSEvSample}
\end{figure}
The problem of reducing features is a combinatorial problem which grows exponentially with the sample size. The backwards approach used by KOMP is a greedy approach which removes only one kernel at a time. As the sample size increases, there are more misleading paths of removal it can take. Additionally, it is only using kernels placed at the sample points, which means it will need more kernels when the true kernel is centered between two sample points.

\subsection{User Localization Problem}\label{Localization}
In the remainder of this section we will apply our method to real world application for which the class of functions the signal belongs to is unknown.
We consider the problem of using RF signals to identify the location of a receiver. Specifically, given the Wi-Fi signal strength from seven routers we wish to identify the room in which our receiver is located \cite{narayanan2016user}. The signal strength varies depending on the location of the router.
 The signal was received from 7 routers spread throughout an office building.  The data was collected using an Android device. At each location the signal strength from each router was observed at $1 s$ intervals. The data was then categorized into 4 groups, each representing the room in which the signal strength was observed. All the rooms are on the same floor, with the rooms representing  the conference room, the kitchen, the indoor sports room, and work
areas \cite{narayanan2016user}. The goal is to be able to accurately detect the location of the android device given the measured signal strength. 

We use 10-fold cross-validation in order to estimate the generalization error of our algorithm as was used in \cite{narayanan2016user}. The dataset was split into 10 sets of equal size with equal distribution of each label. At each turn one of the sets was used for testing while the others were concatenated and used to train the algorithm. This multiclass classification problem was solved using the one-vs-one strategy, which required 6 comparisons. The final class assignment is made through voting. Each comparison makes a prediction on the class of a sample and thus casts a vote for a particular class. The class with the majority of votes is assigned to the sample. The cost function was for this classification problem is

\begin{equation}\label{eqn_class_cost}
c( \bbz, \bby) = \sum_i \max \lbrace 0,1 - y_i \hat{y}_i \rbrace -\bm{\epsilon}. \
\end{equation}
Solving problem \eqref{eqn_the_problem} we obtain an average accuracy of $98\%$, similar to the performance observed in \cite{narayanan2016user}, in which a fuzzy decision tree algorithm with 50 rules was used to obtain an accuracy of $96.65\%$. This result has been observed to be consistent over increasing values of the sparsity parameter $\gamma$.

\subsection{Mnist Digits Classification}\label{MNIST}
We use data of handwritten digits from the MNIST data set \cite{lecun1998mnist}, which consists of a training set of $60,000$ sample-label pairs and a testing set of $10,000$ images and labels. Each sample is a 28-pixel by 28-pixel grayscale image, which was vectorized to form 784 dimensional features. The labels are between 0 and 9 and correspond to the digit written. There are a total of 10 classes.

The number of features is too large to estimate the value of $\alpha(\bbz,w)$ at every $\bbz \in \reals^{784}$. In order to find a set $\ccalX$ over which $\alpha(\bbz,w)$ is defined, we use k-means with $400$ clusters for each digit. Then $\ccalX'$ in \eqref{eqn_w_only} is defined as the set of all cluster centers and the cost function in \eqref{eqn_class_cost} is used. We then run our algorithm using a one-to-one strategy for multi-class classification and achieve an accuracy of $98.12\%$ for an average of 788 features per classification which is comparable to the accuracy found using \eqref{eqn_finite_prob} using the training set as kernel centers and \eqref{eqn_finite_prob} using the centers found through k-means. The complexity of the representation can be further be reduced, however it comes at the cost of the classification accuracy.
\begin{table}[]
\caption{Classification results for \eqref{eqn_finite_prob} using the training samples as kernel centers and using centers selected from k-means and \eqref{eqn_w_only} using the centers selected from k-means}
\begin{tabular}{lll}
 Method & Number of Kernels per Classifier  & Accuracy  \\ \hline
 \eqref{eqn_finite_prob} & 12000 & 98.83 $\%$  \\
\eqref{eqn_finite_prob} with k-means  & 800 & 98.16 $\%$   \\
 \eqref{eqn_w_only} & 788 & 98.12 $\%$ \\ 
 \eqref{eqn_w_only} & 731 & 96.71 $\%$ \\ 
 \eqref{eqn_w_only} & 53 & 85.66 $\%$ \\ 
 \hline
\end{tabular}
\end{table}

Although the dimensionality of the features in the original data makes the use of \eqref{eqn_the_problem} impractical, we can solve that problem, by projecting the data into a lower dimensional space by using principal component analysis (PCA). The formulation in \eqref{eqn_the_problem} has the advantage that the found kernel centers can give some intuition about the distribution of the signal. Particularly, in the case of digits they can describe digits which are representative of written digits. To illustrate that we have performed the classification of the digits '0' and '1' using the first 3 principal components. The low dimensional feature set allows us to find the $\bbx$ which result in the highest value for $\alpha(w,\bbx)$. From these points we can reconstruct the corresponding digits. Figure \ref{fig:typical} shows the resulting images. These are not part of the initial written digit data set but rather represent an image that is closest to all written digit. The accuracy of the classification is $99.62\%$.

\begin{figure}[tb!]
\begin{minipage}[c]{0.45\columnwidth}
  \centering
  \includegraphics[width=\columnwidth]{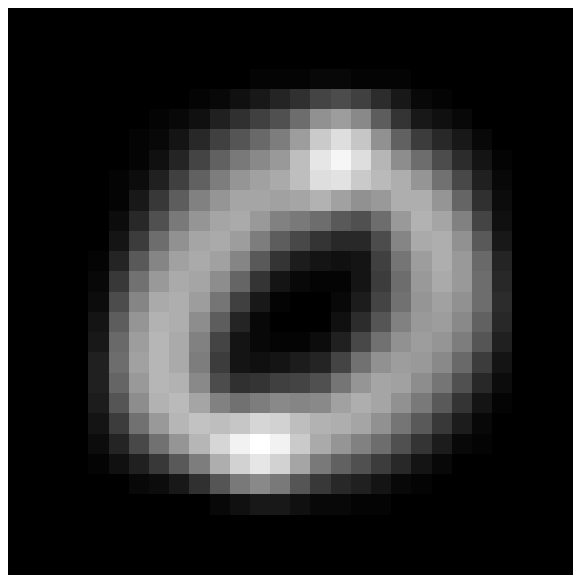}
  {\small (a)}
\end{minipage}
\begin{minipage}[c]{.45\columnwidth}
  \centering
  \includegraphics[width=\columnwidth]{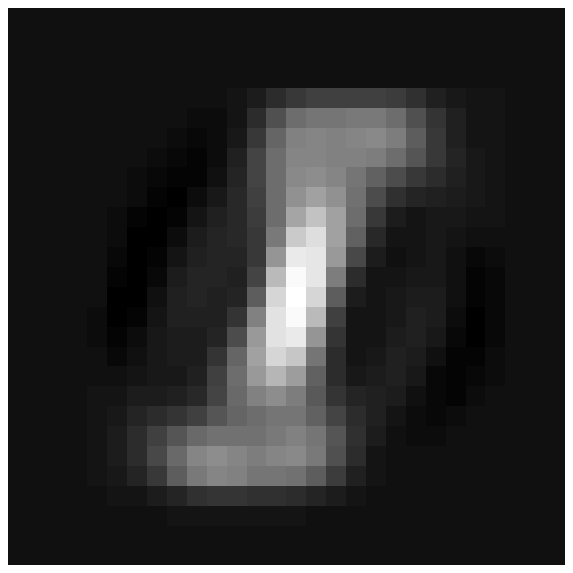}
  {\small (b)}
\end{minipage}
\caption{Kernel centers obtained by solving \eqref{eqn_the_problem} with the highest value for $\alpha(\bbz,w)$ for each digit. These centers are representative of the digits, however, are distinct from any of the samples in the training set. }
\label{fig:typical}
\end{figure}

%%%%%%%%%%%%%%%%%%%%%%%%%%%%%%%
%%% SECTION : Conclusions   %%%
%%%%%%%%%%%%%%%%%%%%%%%%%%%%%%%

\section{Conclusions} \label{sec_conclusions}
In this paper, we have introduced a method for function estimation in RKHS which can model signals of varying degrees of smoothness. The algorithm finds a sparse representation of a  function in the sum space of a family of RKHSs, and determines the kernel parameter for each kernel. Additionally, due to the sparsity, traditional representer theorems no longer hold, so our algorithm placed kernels arbitrarily. While the problem was not convex, a change in the representation of the function to the integral over the product of the coefficient function and the kernel function, allowed us to solve the problem in the dual domain. By leveraging the results on strong duality, we were able to solve the problem in the dual domain.

The theoretical results were validated though simulated signals. We showed that our algorithm finds kernel centers and widths which can represent the function. Furthermore, we showed that our algorithm is able to find sparser representations than KOMP, with the same error. The sparseness of the representation was shown to be independent of sample size, which is not true for greedy kernel reduction methods. 
We also validated our method on a localization dataset, for which we were able to reduce the complexity by $86\%$ while maintaining the high accuracy. Similarly, a sparse kernel representation was obtained for classifying the digits from the MNIST dataset.

\appendices

\section{Proof of Theorem~\ref{thrm_zerodual}}
\label{A:zerodual}

\begin{proof}
In order to show strong duality, it is sufficient to show that the perturbed function $P(\xi)$ in  (\ref{perepi}) is convex \cite{rockafellar2015convex, shapiro2000duality}. Consider the perturbed version of the optimization function

\begin{equation}\label{perepi}
\begin{aligned}[l]
P(\xi) = \underset{\alpha, \bm{\theta}}{\text{min}} & \ f_0(\alpha) \\
 s.t. & \ c(z_i,y_i) \leq \xi_i\\
 &\  z_i =  \int \alpha(\bbz, w) \cdot k(\bbx_i, \bbz; w) \, dw \, d\bbz. \\
\end{aligned}
\end{equation}

In equation (\ref{perepi}) $f_0(\alpha)$ represents the objective function of our original problem: $f_0(\alpha) = \gamma \int \mathbb{I}(\alpha(\bbz,w) \neq 0)  + 0.5\alpha^2(w, \bbx)\, dw \, d\bbz$. The optimal solution for $P(0)$ is the solution to the primal problem.

Convexity of the perturbed problem can be shown, by proving that
given an arbitrary pair of perturbations $\xi_1$ and $\xi_2$ and the corresponding optimal values $P(\xi_1)$ and $P(\xi_2)$,
for any $\beta \in [0,1]$ the solution $P(\xi_\beta)$ has the following property, where $\xi_\beta$ is defined by $\xi_\beta = \beta \xi_1 + (1-\beta) \xi_2 $:

\begin{equation}\label{convex}
P(\xi_\beta) \leq \beta P(\xi_1) + (1-\beta) P(\xi_2).
\end{equation}

In order to prove the convexity of the perturbed problem we need to introduce the following lemma.

\begin{lemma} \label{convexset}
The set of constraints given by 
\begin{equation}
\begin{aligned}
\mathcal{B} = \lbrace b, &  \, \, b = f_0(\alpha), \, \,  c(z_i, y_i)<\xi_i,  \\
& z_i = \int \alpha(\bbz,w) \cdot k(\bbx_i, \bbz; w) \, dw \, d\bbz, ~ ~ i = 1 \cdots N.\rbrace \\
\end{aligned}
\end{equation}
is convex.
\end{lemma}

\begin{proof} Given an arbitrary pair $b_1,b_2 \in \mathcal{B}$, there exists a corresponding $\alpha_1, \alpha_2 \in L_2$ such that $b_1 = f_0(\alpha_1)$ and $b_2 = f_0(\alpha_2)$. In order to prove the convexity of the set, we will show that there exists a feasible $\alpha_\beta \in L_2$ such that for any $\beta \in [0,1]$

\begin{equation}
f_0(\alpha_\beta) = \beta \, b_1 + (1-\beta) b_2
\end{equation}

Let $\mathbb{B}$ be the Borel field of all possible subsets of $\mathcal{U}$, where $\mathcal{U} = \lbrace\mathcal{X} \times \mathcal{W} \rbrace$ is the set of all possible kernel centers and kernel widths and. Let us construct a measure  over $\mathbb{B}$, where $\mathcal{V} \subset \mathbb{B}$. 

\begin{equation}\label{measure}
\bm{m}(\mathcal{V}) = \left[
\begin{matrix}
\int_\mathcal{V} \alpha_1(\bbv) \bbk(\bbv) \, d\bbv \\
 \int_\mathcal{V} \alpha_2(\bbv) \bbk(\bbv) \, d\bbv \\
 \int_\mathcal{V}  \gamma \mathbb{I}(\alpha_1(\bbv) \neq 0) + \alpha_1^2(\bbv) \, d\bbv  \\
   \int_\mathcal{V} \gamma \mathbb{I}(\alpha_2(\bbv) \neq 0) + \alpha_2^2(\bbv) \, d\bbv \\  
\end{matrix} \right]
\end{equation}

The first $2N$ elements of the measure represent the estimated function of the signal $\bby$ using $\alpha_1$ and $\alpha_2$ and a subset of the kernels, where $\bbk(\bbv)_i = k(\bbX,\bbz_v; w_v)$ and $\bbv = [\bbz_v^T,w_v]^T$. 
The last two elements of $\bm{m}(\mathcal{V})$ measure the sparsity of functions $\alpha_1$ and $\alpha_2$ over the set $\mathcal{V}$ respectively. 
Two sets are of interest, the empty set and $\mathcal{U}$. The measure  of the former is $\bm{m}(\emptyset) = 0$ and the measure for $\mathcal{U}$ can be inferred from our optimization problem.

\begin{equation}
\bm{m}(\mathcal{U}) = \left[
\begin{matrix}
 \int_\mathcal{U} \alpha_1(\bbv) \bbk(\bbv) \, d\bbv \\
 \int_\mathcal{U} \alpha_2(\bbv) \bbk(\bbv) \, d\bbv \\
\int_\mathcal{U} \gamma \mathbb{I}(\alpha_1(\bbv) \neq 0) + \frac{1}{2}\alpha_1^2(\bbv) \, d\bbv  \\
\int_\mathcal{U} \gamma \mathbb{I}(\alpha_2(\bbv) \neq 0) + \frac{1}{2} \alpha_2^2(\bbv) \, d\bbv \\  
\end{matrix} \right] = 
\left[
\begin{matrix}
 \hat{\bby}_{1} \\
 \hat{\bby}_{2}  \\
 b_1  \\
  b_2 \\  
\end{matrix} \right] 
\end{equation}

 Lyapunov's convexity theorem \cite{liapounoff1940fonctions} states that a non-atomic measure vector on a Borel field is convex. Note that the representation in \eqref{eqn_integral_rep} allows us to construct the measure with non-atomic masses and is essential to the proof of strong duality. Since $\alpha$ does not contain any point masses our measure $\bm{m}$ is convex . Therefore, for any $\beta \in [0,1]$, there exists a set $\mathcal{V}_\beta \subset \mathbb{B}$ such that:
 
 \begin{equation}\label{Zbeta}
 \bm{m}(\mathcal{V}_\beta) = \beta \, \bm{m}(\mathcal{U}) +(1-\beta) \bm{m}(\emptyset) = \beta \, \bm{m}(\mathcal{U})
 \end{equation}

The measure of the complement of the set $\mathcal{V}_\beta$, as defined by $\mathcal{V}_\beta^c \cup \mathcal{V}_\beta= \mathcal{U} $ and $\mathcal{V}_\beta^c \cap \mathcal{V}_\beta= \emptyset $, can be computed, due to the additivity property of measures:

\begin{equation}\label{complement}
\bm{m}(\mathcal{V}_\beta^c) = \bm{m}(\mathcal{U}) - \bm{m}(\mathcal{V}_\beta) = (1-\beta) \bm{m}(\mathcal{U}) = \bm{m}(\mathcal{V}_{(1-\beta)}).
\end{equation}

We can define the function $\alpha_\beta$ from (\ref{Zbeta}) and (\ref{complement}):

\begin{equation}\label{alphabeta}
\alpha_\beta (\bbv) = \begin{cases}
       \alpha_1(\bbv) &  \bbv \in \mathcal{V}_\beta \\
       \alpha_2(\bbv) & \bbv \in \mathcal{V}_\beta^c \\
     \end{cases}
\end{equation} 

From this construction of $\alpha_\beta (\bbv)$ it can be easily seen that $f_0(\alpha_\beta) = \beta f_0(\alpha_1) + (1-\beta) f_0(\alpha_2)$. Next we will show that $\alpha_\beta$ is feasible. Define $\hat{\bby}_{\beta}$ as:

\begin{equation}
\begin{aligned}
\hat{\bby}_{\beta} & =\int_\mathcal{U} \alpha_\beta(\bbv) \bbk(\bbv) \, d\bbv = \\
& =\int_{\mathcal{V}_\beta} \alpha_1(\bbv) \bbk(\bbv) \, d\bbv  + \int_{\mathcal{V}_\beta^c} \alpha_2(\bbv) \bbk(\bbv) \, d\bbv = \\
& = \beta \hat{\bby}_{1} + (1- \beta) \hat{\bby}_{2} \\
\end{aligned}
\end{equation}

Since $c(z_i,y_i)$ is convex it follows that for any $i \in [1,N]$, $c(\beta z_{i,1} + (1-\beta) z_{i,2}, y_i) \leq \beta c(z_{i,1},y_i) + (1-\beta) c(z_{i,2}, y_i)$. We can use this property to show that $c(z_{i,\beta}, y_i) \leq \xi_\beta$.

\begin{equation}
\begin{aligned}
c(z_{i,\beta}, y_i) \leq \beta c(z_{i,1},y_i) + (1-\beta) c(z_{i,2}, y_i) \leq \\
\leq  \beta \xi_1 + (1-\beta) \xi_2 = \xi_\beta
\end{aligned}
\end{equation}

Thus it was proven that $\alpha_\beta$ is also feasible and therefore the set of constraints is convex.
\end{proof}

Let $(\alpha_1, \bbz_1, \xi_1)$  and $(\alpha_2, \bbz_2, \xi_2)$ be the pair of optimal solutions to the two perturbed problems $P(\xi_1)$ and $P(\xi_2)$. We have shown that there exists a feasible point $\alpha_\beta$ for the problem perturbed by $\xi_\beta = \beta \xi + (1-\beta) \xi'$, which satisfies $f_0(\alpha_\beta) =  P(\xi_1) + (1-\beta) P(\xi_2)$. Given that it is a feasible point the objective function is greater or equal to the solution of the problem

\begin{equation}
\beta P(\xi_1) + (1-\beta) P(\xi_2) = f_0(\alpha_\beta) \geq P(\beta \xi_1+ (1-\beta) \xi_2).
\end{equation}

Since the perturbed problem is convex, the original problem has zero duality gap.

\end{proof}

\section{Proof of Corollary~\ref{T:primal_recovery}}
\label{A:primal_recovery}

\begin{proof}

Theorem~\ref{thrm_zerodual} implies that any solution~$(\alpha^\star, \boldsymbol{\yhat}^\star)$ of~\eqref{eqn_the_problem} is such that~$(\alpha^\star, \boldsymbol{\yhat}^\star) \in \argmin_{\alpha,\,\yhat_i} \ccalL(\alpha,\boldsymbol{\yhat},\bblambda^\star,\bbmu^\star)$~\cite{Boyd04c}. Since~$\ccalL$ in~\eqref{eqn_lagrangian} separates across~$\alpha$ and~$\boldsymbol{\yhat}$, we can consider the minimizations individually to obtain
\begin{equation}\label{eqn_minimax}
	\alpha^\star \in \argmin_{\alpha \in L_2} \ccalL_\alpha(\alpha,\bblambda^\star)
		\text{,}
\end{equation}
for~$\ccalL_\alpha$ from~\eqref{eqn_g_alpha}. We also know from Proposition~\ref{prop_Lagrangian_solution} that~$\alpha_d^\star$ is in the $\argmin$~set of~\eqref{eqn_minimax}. In the sequel, we show that it is~(essentially) its only element.

To do so, we construct~$\alpha^\star$ piece by piece by partitioning the integral in~$\ccalL_\alpha$ into three disjoint sets depending on the value of~$\alpha_d^\star$. On~$\ccalA_{>} = \{(\bbz,w) \in \ccalX \times \ccalW \mid \vert \alpha_d^\star(\bbz,w) \vert > \sqrt{2\gamma}\}$, we know that~$\alpha_d^\star$ takes values from~\eqref{eqn_bar_alpha_star}, the unique minimizer of~$\ccalL_\alpha$ since it stems from the minimization of the strongly convex function~\eqref{eqn_Fprime}. Moreover, our assumptions on the reproducing kernel together with~\eqref{eqn_bar_alpha_star} imply that~$\alpha_d^\star \in L_2$ when restricted to~$\ccalA_{>}$. Hence,
\begin{equation}
	\alpha^\star(\bbz,w) = \alpha_d^\star(\bbz,w)
		\text{,} \quad \text{for } (\bbz,w) \in \ccalA_{>}
		\text{.}
\end{equation}

Over the set~$\ccalA_{<} = \{(\bbz,w) \in \ccalX \times \ccalW \mid \vert \alpha_d^\star(\bbz,w) \vert < \sqrt{2\gamma}\}$, notice from~\eqref{eqn_g_alpha} that the integrand of~$\ccalL_\alpha$ is always non-negative. What is more, it is always strictly positive unless~$\alpha \equiv 0$. This is ready by applying Lemma~\eqref{T:minInt} and the fact that the minimum of~\eqref{eqn_Fprime} is positive. Thus, from the monotonicity of the integral operator, $\alpha^\star$ is again unique and equal to zero on~$\ccalA_{<}$. From~\eqref{alphaopt}, so is~$\alpha_d^\star$ and we obtain
\begin{equation}
	\alpha^\star(\bbz,w) = \alpha_d^\star(\bbz,w) = 0
		\text{,} \quad \text{for } (\bbz,w) \in \ccalA_{<}
		\text{.}
\end{equation}
Immediately, we have that~$\alpha^\star \in L_2$ over~$\ccalA_{>} \cup \ccalA_{<}$.

To conclude the proof, observe that~$\mathfrak{m}\left[ \ccalA_{>} \cup \ccalA_{<} \right] = \mathfrak{m}\left[ \ccalX \times \ccalW \right]$, where~$\mathfrak{m}$ denotes the Lebesgue measure. Indeed, the complement of~$\ccalA_{>} \cup \ccalA_{<}$ is the set~$\ccalA_{=} = \{(\bbz,w) \in \ccalX \times \ccalW \mid \vert \alpha_d^\star(\bbz,w) \vert = \sqrt{2\gamma}\}$. From our assumption on the reproducing kernel, $\ccalA_{=}$ is the set of zeros of a real analytic function, which are isolated and therefore countable~\cite{Krantz02p}. In other words, $\alpha^\star$ and~$\alpha_d^\star$ are in the same equivalence class in~$L_2$ since they are equal except perhaps on a set of measure zero.
\end{proof}

\bibliographystyle{IEEEtran}
\bibliography{myIEEEabrv,bib-mkl,mkl,gsp,math,kernel,sp}

\end{document}